\documentclass{elsarticle}

\usepackage{hyperref}

\journal{\url{https://doi.org/10.1016/j.jspi.2019.01.006} in J. Statist. Plann. Inference; }

\usepackage{array}
\usepackage{epsfig}
\usepackage{fancyheadings}
\usepackage{rotating}

\usepackage{amsmath}
\usepackage{amssymb}
\usepackage{amsfonts}
\usepackage{multirow}
\usepackage{amsthm}
\usepackage[none]{hyphenat}

\newtheorem{theorem}{Theorem}
\newtheorem{lemma}{Lemma}

\newtheorem{definition}{Definition}

\newtheorem*{remark}{Remark}

\fancypagestyle{firststyle}
{
   \fancyhf{sdfdf}
   \fancyfoot[C]{\footnotesize Page \thepage\ of \pageref{LastPage}}
}
\usepackage{color}

\newcommand{\oldnew}{}
\newcommand{\revision}[1]{#1}

\usepackage[all]{xy}
\usepackage{color,bm}
\usepackage{url}

\makeatletter
\newcommand*{\indep}{%
  \mathbin{%
    \mathpalette{\@indep}{}%
  }%
}
\newcommand*{\nindep}{%
  \mathbin{
    \mathpalette{\@indep}{\not}
  }%
}
\newcommand*{\@indep}[2]{%
  \sbox0{$#1\perp\m@th$}
  \sbox2{$#1=$}
  \sbox4{$#1\vcenter{}$}
  \rlap{\copy0}
  \dimen@=\dimexpr\ht2-\ht4-.2pt\relax
  \kern\dimen@
  {#2}%
  \kern\dimen@
  \copy0 
}
\makeatother
\newtheorem{assumption}{Assumption}
\usepackage{subfigure}
\newcommand*{\QEDB}{\hfill\ensuremath{\square}}
\def\d{\mathrm{d}}

 \usepackage{tikz}
\usetikzlibrary{shapes,arrows,positioning,calc}
\graphicspath{ {./Pictures/} }

\def\T{\top}

 \usepackage{tikz}
\usetikzlibrary{shapes,arrows,positioning,calc}
\renewcommand{\thefootnote}{\arabic{footnote}}

\let\svthefootnote\thefootnote

\def\d{\mathrm{d}}

\newcommand\norm[1]{\lVert#1\rVert}
\newcommand{\vertiii}[1]{{\left\vert\kern-0.25ex\left\vert\kern-0.25ex\left\vert #1
    \right\vert\kern-0.25ex\right\vert\kern-0.25ex\right\vert}}
\newcommand{\argmin}{\arg\!\min}

\newcommand{\truepara}{{\beta^*}}
\newcommand{\trueparas}{{\beta^*_{S}}}
\newcommand{\trueparasc}{{\beta^*_{{S}^c}}}
\newcommand{\trueparaj}{{\beta^*_j}}
\newcommand{\estipara}{{\hat{\beta}}}
\newcommand{\estiparaj}{{\hat{\beta}_j}}
\newcommand{\para}{{\beta}}

\newcommand{\estiparas}{{\hat{\beta}_{S}}}

\newcommand{\tuning}{{\lambda}}
\newcommand{\R}{{\mathbb{R}}}
\newcommand{\primal}{{\hat{\alpha}}}
\newcommand{\prims}{{\hat{\alpha}_{S}}}
\newcommand{\primsc}{{\hat{\alpha}_{{S}^c}}}
\newcommand{\dual}{{\hat{\nu}}}
\newcommand{\duals}{{\hat{\nu}_{S}}}
\newcommand{\dualsc}{{\hat{\nu}_{{S}^c}}}
\newcommand{\dualj}{{\hat{\nu}_j}}
\newcommand{\deriva}{{W}}
\newcommand{\derivafun}{w}

\def\d{\mathrm{d}}
\newcommand{\res}{r}
\newcommand{\ratio}{a}
\newcommand{\redef}{=}

\newcommand{\cmax}{\ensuremath{c_{\max}}}

\newcommand{\constA}{C}

 \def\AIC{\textsc{aic}}
 \def\BIC{\textsc{bic}}
 \def\EBIC{\textsc{ebic}}
 \def\GIC{\textsc{gic}}
 \def\CV{\textsc{cv}}

 \makeatletter
 \newcommand{\opnorm}{\@ifstar\@opnorms\@opnorm}
\newcommand{\@opnorms}[1]{%
  \left|\mkern-1.5mu\left|\mkern-1.5mu\left|
   #1
  \right|\mkern-1.5mu\right|\mkern-1.5mu\right|
}
\newcommand{\@opnorm}[2][]{%
  \mathopen{#1|\mkern-1.5mu#1|\mkern-1.5mu#1|}
  #2
  \mathclose{#1|\mkern-1.5mu#1|\mkern-1.5mu#1|}
}
\makeatother

\begin{document}

\begin{frontmatter}

\title{Tuning parameter calibration for\\$\ell_1$-regularized logistic regression}

\author{Wei Li
}
\address{School of Mathematical Sciences, Peking University, Beijing,
China}

\author{Johannes Lederer\corref{mycorrespondingauthor}}
\cortext[mycorrespondingauthor]{Corresponding author}
\ead{ledererj@uw.edu}
\address{Departments of Statistics and Biostatistics, University of Washington, Seattle, WA, USA}

\begin{abstract}
Feature selection is a standard approach to understanding and modeling high-
dimensional classification data, but the corresponding statistical methods hinge on tuning parameters that are difficult to calibrate. In particular, existing calibration schemes in the logistic regression framework lack any finite sample guarantees. In this paper, we introduce a novel calibration scheme for~$\ell_1$-penalized logistic regression. It is based on simple tests along the tuning parameter path and is equipped with optimal  guarantees for feature selection. It is also amenable to easy and efficient implementations, and it rivals or outmatches existing methods in simulations and real data applications.
\end{abstract}


\end{frontmatter}
\pagestyle{empty}

\def\thefigure{\arabic{figure}}
\def\thetable{\arabic{table}}

\pagestyle{plain}
\def\n{\noindent}
\lhead[\fancyplain{} \leftmark]{}
\chead[]{}
\rhead[]{\fancyplain{}\rightmark}
\cfoot{}

\section{Introduction}
\label{sec:intro}

The advent of high-throughput technology has created a large demand for feature selection with high-dimensional classification data.
In gene expression analysis or genome-wide association studies, for example, investigators attempt to select from a large set of potential risk factors the predictors that are most useful in discriminating two or more conditions of interest.  The standard approaches for such tasks  are penalized likelihood methods \citep{buhlmann2011statistics,hastie2015statistical, bunea2008honest,ravikumar2010high,ryali2010sparse,Wu09}.
However, the performance of these methods hinges on the calibration of tuning parameters that balance  model fit and model complexity.

The focus of this paper is the calibration of the~$\ell_1$-penalized likelihood for feature selection in logistic regression.
 The most widely used schemes for this calibration are based on Cross-Validation~(\CV) \citep{stone1974cross} or information criteria, including the Akaike's information criterion (\AIC) \citep{akaike1973information}, the Bayesian information criterion (\BIC) \citep{schwarz1978estimating}, {\oldnew the extended Bayesian information criterion (\EBIC) \citep{chen2012extended},} {and the Generalized Information Criterion (\GIC) \citep{nishii1984asymptotic}}.
 \CV- and~\AIC-based procedures are designed for prediction and  thus typically not suited for feature selection~\citep{shao1993linear}.  In contrast, \BIC-{\oldnew, \EBIC-} and~\GIC-type procedures, see also the recent versions in \citep{wang2007tuning,chen2012extended},  are designed primarily for feature selection, and some consistency results for model selection have been derived~\citep{Fan13,Zhang10}.  Yet another approach based on a permutation idea has been introduced in~\citep{sabourin2015permutation}. However, all these methods share the same limitation in that they lack finite sample guarantees. This means that theoretical backup for applications, where samples sizes are always finite, is not available.

 In this paper, we introduce a novel calibration scheme based on a testing procedure and  sharp~$\ell_{\infty}$-bounds. It is easy to implement and computationally efficient, and in contrast to all previous approaches, it is indeed equipped with finite sample guarantees. Our proposal is thus of immediate practical and theoretical relevance.

The remainder of this paper is organized as follows. Section~\ref{sec:meth} contains our main proposal and the theoretical results. Section~\ref{sec:sim} and Section~\ref{sec:rd} demonstrate that our method is also a contender in simulations and real data applications. Section~\ref{sec:dis} contains  a brief discussion. The proofs and further simulations are deferred to the Appendix.

\vspace{5mm}
\noindent\textit{Notation.}\quad The index sets are denoted by $[k]=\{1,\ldots,k\}$ for~$k\in\{1,2,\dots\}$, and the cardinality of sets is denoted by $|\cdot |$.  For a given vector~$\beta\in \R^p$, the support set of~$\beta$ is written as~$\textnormal{supp}(\beta) \redef \{j\in [p]: \beta_j \neq 0\}$, and  for~$q\in[1,\infty]$, the~$\ell_q$-norm of~$\beta$ is denoted by~$\norm{\beta}_{q}$. The~$\ell_q$-induced matrix-operator norms are denoted by~$\opnorm{\cdot}_{q}$. Two examples are the spectral norm~$\opnorm{\cdot}_2$, which denotes the maximal singular value of a matrix, and the~$\ell_{\infty}$-matrix norm~$\opnorm{X}_{\infty} = \max_{i=1,\ldots,n}\sum_{j=1}^{p}|X_{ij}|$. The minimal and the maximal eigenvalue of a square matrix are denoted by~$\Omega_{\min}(\cdot)$ and~$\Omega_{\max}(\cdot)$, respectively. For a given subset~$A$ of~$[p]$, the vectors~$\beta_A\in \R^{\mid A\mid}$ and~$\beta_{A^c}\in \R^{\mid A^c \mid}$ denote the components  of~$\beta$ in~$A$ and in its complement~$A^c$, respectively, and given a matrix~$X\in \R^{n\times p}$, the matrix~$X_{A}$ denotes the sub-matrix of~$X$ with column indexes restricted to~$A$. The diagonal matrix with diagonal elements~$a_1,\ldots,a_n$ is denoted by~$\textnormal{diag}\{a_1,\ldots,a_n\}$.  The function~$\derivafun$ is finally defined as~$\derivafun(u,v) = \exp(u^\T v)/
(1+\exp(u^\T v))^2 $ for vectors~$u,v$ of the same length.

\section{Methodology}\label{sec:meth}

\subsection{Model and assumptions}

In this section, we formulate the general setting and introduce the assumptions required for the theoretical analysis.
We consider data in the form of a real-valued $n \times p $ design matrix~$X$ and a binary response vector~$Y = (y_1, \ldots, y_n)^\T$.  Our framework allows for high-dimensional data, where~$p$ rivals or outmatches~$n$. We denote the rows of~$X$ (i.e. the samples) by $x_{1}$, \ldots, $x_{n} \in \R^p$ and the columns of~$X$ (i.e. the predictors) by $x^{1}$, \ldots, $x^{p}\in \R^n$. The matrix $X$ can be deterministic or a realization of  a random matrix, {\oldnew but in any case, we assume that $X$ has been normalized, i.e., $\norm{x^j}_2=1$ for $j=1,\ldots,p$.
}

The design matrix $X$ and the response vector $Y$ are linked by the standard logistic regression model
\begin{equation}\label{eqn:LR}
 \Pr(y_i=1\mid x_{i})=\frac{\exp(x_{i}^\T\truepara)}{1+\exp(x_{i}^\T\truepara)} \qquad (i=1,\ldots,n)\,,
\end{equation}
 where $\truepara\in \R^p$ is the unknown regression vector. Our goal is feature selection (or also called support recovery), that is,  estimation of the support set $S\redef\textnormal{supp}(\truepara)$.  The starting point for approaching this task is the well-known family of estimators
\begin{equation}\label{eqn:SLR}
  \estipara_{\tuning}\in \argmin_{\beta\in \R^p}\left\{L(\para) + \tuning\norm{\beta}_1\right\}\qquad (\tuning>0)
\end{equation}
indexed by the tuning parameter $\tuning.$ The first term $L(\para) \redef \sum_{i=1}^n(\log(1+\exp(x_{i}^\T
\beta))-y_ix_{i}^\T\beta)/n$ is the negative log-likelihood function, and the second term is a regularization that exploits  that  $s=|S|\ll n,p$  in many applications. We estimate~$S$ by $\textnormal{supp}( \estipara_{\tuning})$ for a data-driven tuning parameter $\tuning\equiv\tuning(Y,X)$.

\revision{We take~\eqref{eqn:SLR} as our starting point, because it has become the most popular family of estimators in our context.
The main reason for this popularity is that $\ell_1$-regularization is equipped with fast algorithms and some theoretical understanding.
One might argue that different starting points might suit some task better;
for example, one drawback of $\ell_1$-regularization is that it leads to complex dependencies on the design matrix and can require potentially unrealistic assumptions \cite{dalalyan2017prediction,hebiri2013correlations,lederer2016oracle,van2013lasso,zhuang2018maximum}.
But, importantly, our calibration approach is agnostic to what estimator is used:
it only requires a suitable oracle inequality (while deriving such oracle inequalities might be difficult, of course).
In this sense, the following results are merely an indication of the full potential of our calibration scheme.
}

Support recovery \revision{with~\eqref{eqn:SLR} }  is feasible only if the correlations in the design matrix~$X$ are sufficiently small. In the following, we state corresponding assumptions that virtually coincide with those ones used by~\cite{ravikumar2010high} in the context of Ising models. The assumptions are formulated in terms of {\oldnew $X^\top\deriva X/n$, the Hessian of the log-likelihood function evaluated at the true regression parameter $\truepara$, where $\deriva\redef \textnormal{diag}\{\derivafun(x_{1},\truepara),\ldots,\derivafun(x_{n},\truepara)\}$.} We first require that the submatrix of the Hessian matrix corresponding to the relevant covariates has eigenvalues bounded away from zero.
 \begin{assumption}[Minimal eigenvalue condition]
   \label{Ass:BEC} It holds that
   \[
   c_{\min}=\Omega_{\min}(X_{S}^\T\deriva X_{S}/n)>0\, .
   \]
 \end{assumption}
\noindent Note that if this assumption were violated,  the relevant covariates would be linearly dependent, and the true support set~$S$ would not be well-defined. {\oldnew Thus, this condition ensures that the relevant covariates do not become highly dependent.}
We {\oldnew then} impose an irrepresentability condition~\citep{zhao2006model}.
 \begin{assumption}[Irrepresentability condition]
  \label{Ass:MIC} It holds that
\[\gamma=1- \opnorm{(X_S^\T \deriva X_{S})^{-1}X_{S}^\T \deriva X_{S^c}}_{\infty}>0\,.\]
\end{assumption}
\noindent This assumption is a modified version of the irrepresentability condition commonly used in the theory for linear regression with the Lasso~\citep{zhao2006model}.
More generally, irrepresentability conditions prevent the relevant covariates from being strongly correlated with the irrelevant covariates. This  ensures that the true support set can be identified with finitely many samples.


{\oldnew Assumption~\ref{Ass:BEC} and~\ref{Ass:MIC} also make assumptions on $s=|S|$ implicit.
For illustration, suppose that the two assumptions hold for the population matrices $E_{\beta^*}(X_S^{\top} WX_S^{\top})$ and $\big\{E_{\beta^*}(X_S^{\top} WX_S^{\top})\big\}^{-1}E_{\beta^*}(X_S^{\top} WX_{S^c}^{\top})$. Then, according to Lemma 5 and 6 in \cite{ravikumar2010high}, these two assumptions hold with large probability if $\log(p)s^3=o(n)$. In this spirit, larger $s$ makes the conditions more restrictive on other aspects of the model.}

Importantly, however, the above assumptions on the design are not needed in the analysis of the proposed scheme itself. Instead, the assumptions are needed to ensure that there is a viable estimator in the family~\eqref{eqn:SLR} at all. We discuss this in the following section.

\subsection{$\ell_{\infty}$-estimation and support recovery}\label{subsec:boundandsr}

$\ell_{\infty}$-estimation and support recovery  are two closely related aspects of high-dimensional logistic regression. In this section, we thus establish oracle inequalities for both these tasks.

To state the result, we define the vector of residuals as $\varepsilon = (\varepsilon_1,\ldots,\varepsilon_n)^\T$ with entries  $\varepsilon_i\redef y_i-\Pr(y_i=1\mid x_{i})$ for $i\in [n]$. The vector $\varepsilon$ is random noise with mean zero.
{\oldnew As it is standard in the theory of high-dimensional statistics, our results are based on an event
\[
\mathcal{T}_{\tuning}\redef\Big\{\frac{4(2-\gamma)}{n\gamma}\norm{X^\T \varepsilon}_{\infty} \leq \tuning\Big\}\qquad (\tuning>0)\,.
\]
Terms such as $\frac{4(2-\gamma)}{n\gamma}\norm{X^\T \varepsilon}_{\infty}$ are sometimes called the ``effective noise''~\cite{wainwright2014structured} and can often be bounded by standard empirical process theory~\cite{buhlmann2011statistics}.
Essentially, the event states that heavy noise ($\varepsilon$ large)  requires strong penalization ($\tuning$ large).
\addtocounter{footnote}{0}\let\thefootnote\svthefootnote
For the technical proofs, we also
 assume $\tuning\leq \gamma c_{\min}^2/(100(2-\gamma)s\cmax)$ in the remainder,\footnote{On a high level, $\gamma c_{\min}^2/(100(2-\gamma)s\cmax)\sim 1$  and ${4(2-\gamma)}\norm{X^\T \varepsilon}_{\infty}/(n\gamma)\sim 1/\sqrt n $. Thus, $\gamma c_{\min}^2/(100(2-\gamma)s\cmax)\gg {4(2-\gamma)}\norm{X^\T \varepsilon}_{\infty}/(n\gamma).$ Since the right-hand side of this relation is basically the optimal tuning parameter targeted in our study, see the next section, the much larger upper bound on $\tuning$ has no impact on our analysis. For details, in particular on the constants, we refer to the proofs section.}
 and for ease of presentation, we  set}
 \[
     \ratio=\opnorm{(X_{S}^\T\deriva X_{S})^{-1}}_{\infty}/
\opnorm{(X_{S}^\T\deriva X_{S})^{-1}}_{2}\,.
\]
 Then, we find the following result.
\begin{theorem}[$\ell_{\infty}$-bound and support recovery]
 \label{Thm:SR}
  Under Assumption~\ref{Ass:BEC} and~\ref{Ass:MIC},
  the following properties hold on the event $\mathcal{T}_{\tuning}$.
  \begin{enumerate}
    \item[(a)] $\ell_{\infty}$-bound: $\norm{\estipara_{\tuning}-\truepara}_{\infty}\leq 1.5\ratio\tuning/c_{\min}$\,;
    \item[(b)] support recovery: $\textnormal{supp}(\estipara_{\tuning})\subset S$, and $\textnormal{supp}(\estipara_{\tuning}) = S$ if $\min_{j \in S}|\trueparaj|>1.5\ratio\tuning/c_{\min}$\,.
  \end{enumerate}
\end{theorem}
\noindent Oracle inequalities are the standard way to state finite sample bounds in high-dimensional statistics~\citep{buhlmann2011statistics}. Similar results for $\ell_1$-penalized logistic/linear regression have also been derived elsewhere~\citep{bunea2008honest,chichignoud2016practical,Karim08,ravikumar2010high}, but the above formulation is particularly useful for our purposes.  Part~(a) implies that for a suitable tuning parameter~$\tuning$, the  estimator~$\estipara_{\tuning}$ is uniformly close to the regression vector~$\truepara$. Part~(b) implies that for suitable tuning parameter, the estimator $\textnormal{supp}(\estipara_{\tuning})$ provides exact support recovery if the non-zero parameters are sufficiently large. As long as the design assumptions are met, Theorem~\ref{Thm:SR} thus ensures that the family~\eqref{eqn:SLR} contains a viable estimator.

\subsection{Testing-based calibration}\label{subsec:AV}

Theorem~\ref{Thm:SR} ensures that the family~\eqref{eqn:SLR} contains an accurate estimator. This leaves us with two  tasks: (i)~We have to formulate a {\oldnew notion} of optimality within the family~\eqref{eqn:SLR}. In other words, we have to define what an optimal tuning parameter is. (ii)~We have to formulate a scheme to find an optimal tuning parameter from data.

To address these two tasks, we develop an approach that relates to the AV-testing idea introduced by~\cite{chichignoud2016practical}. The AV-tests have been developed for linear regression, which differs from logistic regression both in theory and implementations. For example, $\ell_\infty$-bounds in linear regression can be established by ``standard'' proof techniques, while the bounds needed here are based on the more recent Primal-Dual Witness technique. However, a more interesting, and quite striking insight here is that  the high-level arguments transfer from the linear to the non-linear setting - and even beyond. The following discussion can thus be read as a general blueprint for feature selection calibration, while the parts specific to logistic regression, namely the proof techniques and details, are deferred to the Appendix.

Let us first define the concept of oracle tuning parameters. Since one can handle only finitely many values in practice, we consider a fixed but arbitrary  sequence $ 0< \tuning_1 < \cdots < \tuning_N$ of tuning parameters and denote  the corresponding set by $\Lambda\redef \{\tuning_1,\cdots,\tuning_N\}$.
In view of Theorem~\ref{Thm:SR}, an optimal tuning parameter satisfies two requirements.  On the one hand, the bounds hold only on the  event~$\mathcal{T}_{\tuning}$. Thus, an optimal tuning parameter needs to ensure that the event~$\mathcal{T}_{\tuning}$ holds with high probability. On the other hand, the bounds are linear in~$\tuning$. Thus, an optimal tuning parameter should be as small as possible.  We formalize this notion as follows:
      \begin{definition}[Oracle tuning parameter]
     Given $\delta\in (0,1)$, the oracle tuning parameter is
     \begin{equation}
       \tuning_{\delta}^*\redef \argmin_{\tuning \in \Lambda}\{\Pr(\mathcal{T}_{\tuning})\geq 1-\delta\}\,.
     \end{equation}
   \end{definition}
\noindent Since the set $\Lambda$ is finite, the oracle tuning parameter is always well-defined. {\oldnew It is also small in large samples: in particular, since the residuals $\varepsilon$ in $\mathcal{T}_{\tuning}$ are bounded, standard concentration results ensure that  $\lambda^*_\delta\to 0$ for $n\to\infty$ and $\delta$ fixed---as long as $\log p/n\to 0$ for $n\to\infty$.}

We call the optimal tuning parameter ``oracle tuning parameter'' to signify that it is  a purely theoretical quantity and  cannot be used in applications. First, $\tuning_{\delta}^*$ depends on $\gamma,$ which is unknown in practice. Second, even if $\gamma$ were known, a precise evaluation of $\tuning_{\delta}^*$ would be computationally intensive. Finally, it is unclear how to choose $\delta.$  We thus aim at finding a data-driven selection rule that mimics the performance of the optimal tuning parameter. The following tests  provide this.
   \begin{definition}[Testing-based calibration]
     Given a constant $\constA\geq 1.5\ratio/ c_{\min}$, we select the tuning parameter
     \begin{equation}\label{eqn:AV}
     \hat{\tuning}\redef \min\Big\{\tuning\in \Lambda\ :\
     {\norm{\estipara_{\tuning'}-\estipara_{\tuning''}}_{\infty}}\leq \constA\tuning'+\constA\tuning''~~\forall\, \tuning',\tuning''\geq \tuning \Big\}
     \end{equation}
and set
\begin{equation}\label{eset}
  \hat S=\{j\in[p]:|(\estipara_{\hat\tuning})_j|\geq 3\constA\hat\tuning\}\,.
\end{equation}
   \end{definition}
\noindent {\oldnew An intuition goes as follows:
  suppose that $\tuning$ is large enough to control the random noise, that is, ${4(2-\gamma)}\norm{X^\T \varepsilon}_{\infty}/({n\gamma}) \leq \tuning$. Then, also $\tuning',\tuning''\geq \tuning$ are large enough, and Theorem~\ref{Thm:SR} ensures that $\norm{\estipara_{\tuning'}-\truepara}_{\infty}\leq \constA\tuning'$ and $\norm{\estipara_{\tuning''}-\truepara}_{\infty}\leq \constA\tuning''$.
  Combining these two inequalities with the help of the triangle inequality shows that $\norm{\estipara_{\tuning'}-\estipara_{\tuning''}}_{\infty}\leq \constA (\tuning'+\tuning'')$ is a necessary condition for $\lambda$ being large enough.
  What we now want is the smallest one among such ``large enough'' tuning parameters.
  This motivates us to select $\hat\tuning$ as the smallest $\tuning$ that satisfies $\norm{\estipara_{\tuning'}-\estipara_{\tuning''}}_{\infty}\leq \constA (\tuning'+\tuning'')$ for all $\tuning',\tuning''\geq \tuning$.
  This selection is then indeed ``conservative'' ($\hat\tuning\leq \tuning_{\delta}^*$), but it is also optimal in the sense of  Statement~\eqref{eset}.
}

Two features of our testing-based scheme are apparent immediately: First, the method is computationally efficient, because it requires at most one pass of the tuning parameter path. This path can be computed by standard algorithms such as~\texttt{glmnet}~\citep{Friedman09}. Since the structure of the tests allows for early stopping, the computation of a part of the tuning parameter path is actually sufficient. Second, the method is easy to implement because it consists of simple $\ell_\infty$-tests along the tuning parameter path. The tests also highlight the close connections between $\ell_\infty$-estimation and our final goal, support recovery.

The third feature of our scheme is that it is equipped with optimal finite sample theoretical guarantees. We establish this in the following result.
   \begin{theorem}[Optimality of the testing-based calibration]
   \label{Thm:AV}
     Under Assumption~\ref{Ass:BEC} and \ref{Ass:MIC}, for any $\delta\in (0,1)$ and $\constA\geq 1.5 \ratio/c_{\min}$, the tuning parameter~$\hat{\tuning}$ from  \eqref{eqn:AV} provides with probability at least $1-\delta$
 \[\hat{\tuning}\leq \tuning_{\delta}^* \quad \text{and} \quad
 \norm{\estipara_{\hat{\tuning}}-\truepara}_{\infty}\leq 3\constA\tuning_{\delta}^*\,,
\]
and, if $\min_{j \in S}|\trueparaj|>6\constA\tuning_{\delta}^*,$
     \begin{equation*}
       \hat S \supset S\,.
     \end{equation*}
   \end{theorem}
\noindent Let us highlight some aspects of this result: First, all results are stated for fixed $n,p$, and all constants are specified. The bounds are thus finite sample bounds that can provide, as opposed to asymptotic bounds, concrete insights into the practical performance of the method.  Next, the guarantees hold for any $\gamma$ and~$\delta,$  but these quantities do  {\it not} need to be specified in the method. Similarly, the results hold irrespective of the set~$\Lambda,$ in particular, irrespective of the number of tuning parameters~$N$. The set~$\Lambda$ enters the results only  through~$\tuning_{\delta}^*$: the finer the grid~$\Lambda$, the more precise the optimal tuning parameter~$\tuning_{\delta}^*, $ and thus, the sharper the guarantees. Furthermore, the $\ell_\infty$-bounds demonstrate the optimality of the method. Indeed,  the estimator with optimal, in practice unknown tuning parameter satisfies $\norm{\estipara_{\tuning_{\delta}^*}-\truepara}_{\infty}\leq 1.5\ratio\tuning_{\delta}^*/c_{\min}$, see Theorem~\ref{Thm:SR}. The bound for the estimator with the data-driven tuning parameter~$\hat\tuning$ equals this bound - up to a constant factor~3. Finally, since Definition~\eqref{eset} contains a threshold, which is based on the guarantee $\hat{\tuning}\leq \tuning_{\delta}^*,$  the number of false positives is typically small. Yet, the second part of the theorem ensures that $\hat S$ contains all sufficiently large predictors, which means that also the number of false negatives is typically small.  Theorem~\ref{Thm:AV} thus provides accurate feature selection guarantees for the testing-scheme.  We are not aware of any comparable feature selection (or $\ell_\infty$-) guarantees for standard calibration schemes.

{\oldnew
\begin{remark}[The constant $\constA$ in practice]
 \revision{The optimal value of $\constA$ in view of the theoretical bounds is $\constA=1.5\ratio/ c_{\min}$.
As described above, support recovery with~\eqref{eqn:SLR} is not possible in highly correlated settings, and it has been pointed out that large $\beta^*$ can be problematic for $\ell_1$-penalized methods more generally~\citep{dalalyan2017prediction}.
Assuming near-orthogonal design and small parameter values in the sense of $\norm{\beta^*}_2\approx 0$, so that $\ratio\approx 1$, $c_{\min}\approx 1/4$, we find  $1.5 \ratio/c_{\min}\approx 6$.
This suggests that an appropriate choice is $\constA=6$, and we adopt this choice throughout this paper.

The assumption of near-orthogonal design and small parameter values might be unrealistic in practice,
but the empirical studies in Section~\ref{sec:sim} and the Appendix demonstrate good performance of $\constA=6$ even when the model deviates substantially from this assumption.

In any case, the choice of $\constA$ remains a subject for further study.
For example, the limitations of the heuristics for $\constA$  might simply be an artifact of using the $\ell_1$-regularizer:
there might be estimators different from~\eqref{eqn:SLR} for which the application of our calibration scheme leads to constants~$\constA$ that can be theorized more globally.
}
\end{remark}
To summarize, the proposed testing-based method accurately mimics the performance of the optimal tuning parameter, and yet, it is computationally efficient and does not depend on the quantities~$\gamma$ and~$\delta.$ Moreover, the parameter~$\constA$ can be set to a universal constant; in particular, $\constA$ does not require calibration. The simulation results below indicate that indeed no further calibration is required. The testing-based scheme is thus a practical scheme with a sound theoretical foundation.
}
\section{Simulation studies}\label{sec:sim}

In this section, we show the practical performance of the proposed scheme in a simulation study. We simulate data from the logistic regression model~\eqref{eqn:LR} with $n = 200$ samples and $p\in\{200, 500\}$ predictors. The row vectors $x_{i}$ of the design matrix $X$ are i.i.d. Gaussian  with mean zero and covariance $\Sigma = (1-\kappa)\textnormal{I} + \kappa \textbf{1}$, where $\textnormal{I}$ is the identity matrix, $\textbf{1}$ is the matrix of all ones, and $\kappa\in \{0, 0.25, 0.5, 0.75\}$ the level of correlations among the predictors. The settings with larger correlations are particularly interesting, because they violate the strict design assumptions. The coordinates of the regression vector~$\truepara$ are set to zero except for $s\in \{8,12,15{\oldnew , 20, 25}\}$ uniformly at random chosen entries that are set to $1$ or $-1$ with equal probability. While our theory holds for any tuning parameter grid, we consider $N=500$ tuning parameters that are equally spaced on $[\tuning_{1},\tuning_{N}]$, where $\tuning_{1}=0.0001\tuning_{N}$ and $\tuning_{N}=10\log(p)/n$ ensure a large spread of outcomes. For each of the total {\oldnew $40$} settings,  we report the means over  $200$~replications. The methods under consideration are the testing-based method defined in~\eqref{eqn:AV} and~\eqref{eset},  \BIC, {\oldnew \EBIC}, 10-fold \CV, and \AIC. {\oldnew The \EBIC~is the
classical \BIC~with an additional penalty term $2\theta\log(p)$ with a positive $\theta$. Here, we choose $\theta\in\{0.25,0.5,1\}$. Note that the \BIC~is a special form of \EBIC~with $\theta=0$.}
No thresholding is applied for the standard methods, since there is no guidance on the choice of such a threshold. All computations are conducted with the software \texttt{R}~\citep{Rsoftware} and the \texttt{glmnet} package.

Since our goal is support recovery, we compare the methods in terms of  Hamming distance, which is the sum of {the number of} false positives and false negatives. Figure~\ref{fig:hdf} contains the results for {\oldnew $\kappa\in\{0.25,0.5\}$ and $s\in\{8,12,15\}$}.
\begin{figure}[h!]
\centering
       \subfigure[~ $p = 200$, $\kappa = 0.25$]{
  \includegraphics[width=0.45\textwidth]{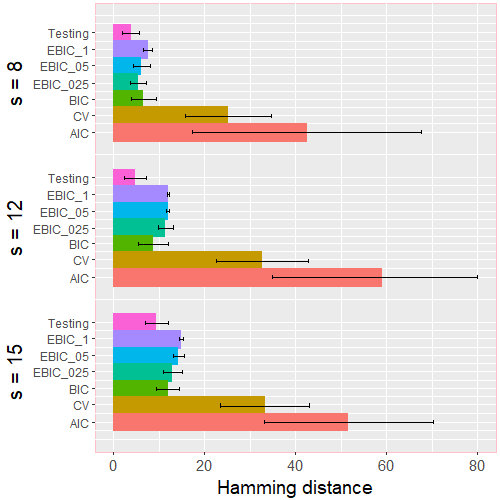}
  }
  ~~~~
\subfigure[~ $p = 200$, $\kappa = 0.5$]{
  \includegraphics[width=0.45\textwidth]{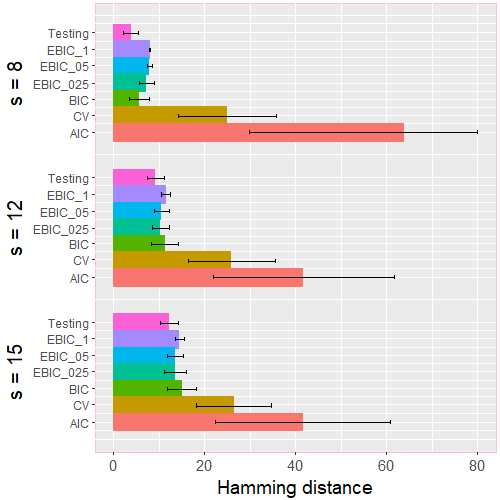}
  }
  \vspace{6mm}
 \subfigure[~ $p = 500$, $\kappa = 0.25$]{{
  \includegraphics[width=0.45\textwidth]{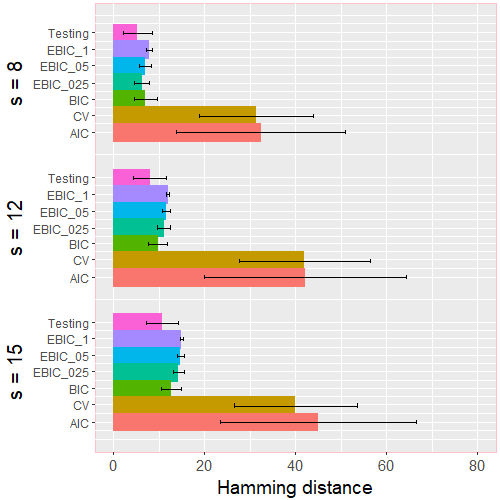}
  }}
  ~~~~
  \subfigure[~ $p = 500$, $\kappa = 0.5$]{{
  \includegraphics[width=0.45\textwidth]{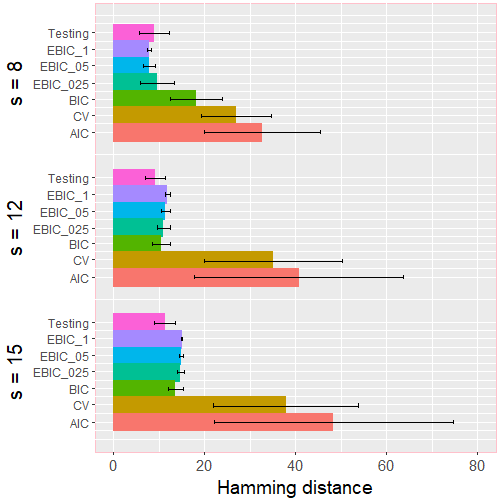}
  }}
  \caption{Variable selection errors of $\ell_1$-regularized logistic regression with {\oldnew seven} different calibration schemes for the tuning parameter. The $12$ simulation settings differ in the number of variables $p$, correlation $\kappa$, and sparsity level~$s$.}\label{fig:hdf}
\end{figure}
{\oldnew In each plot of Figure~\ref{fig:hdf}, we use $\EBIC\_1$, $\EBIC\_05$ and  $\EBIC\_025$ to denote
 the EBIC method with $\theta=1,0.5$ and 0.25, respectively.}
 The results for the other correlation {\oldnew and sparsity} levels are deferred to the Appendix. The results allow for two observations: First, \BIC\ {\oldnew and \EBIC}\ consistently outperform \CV\ and \AIC. This is no surprise, given that \BIC\ {\oldnew and \EBIC\ are} specifically designed for feature selection. Second, our testing-based scheme rivals \BIC\ {\oldnew and \EBIC}\ across all settings.

\BIC, {\oldnew  \EBIC}\ and \AIC\ require one complete pass of the tuning parameter path. 10-fold \CV\ requires one complete pass of 10 tuning parameter paths and thus, requires about 10 times more computational power (or parallelization). The testing-based scheme is the most efficient approach: it requires at most one complete pass of the tuning parameter path, and typically even less, because it stops  as soon as the tuning parameter is selected. For illustration, Figure~\ref{fig:timef} summarizes the run times for six settings with $\kappa=0.5$; additional results with larger correlations {\oldnew and more dense signals} are provided in the Appendix.
 \begin{figure}
\centering
       \subfigure[~ $p = 200$, $\kappa = 0.5$]{
  \includegraphics[width=0.45\textwidth]{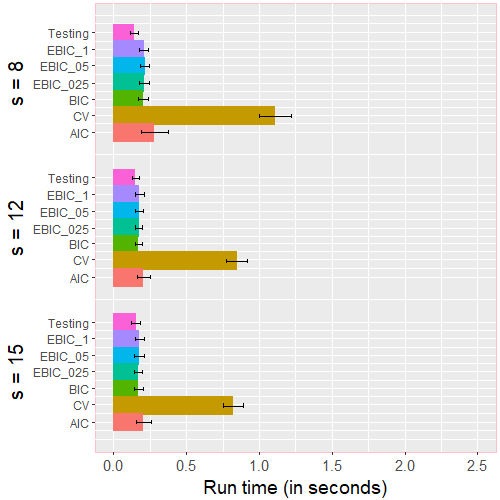}
  }
  ~~~~
\subfigure[~ $p = 500$, $\kappa = 0.5$]{
  \includegraphics[width=0.45\textwidth]{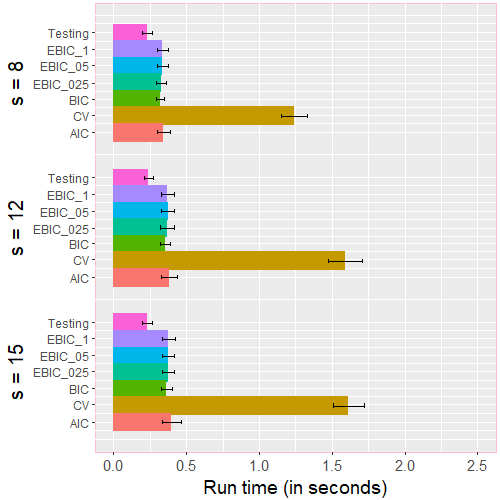}
  }
  \caption{Run times (in seconds) of $\ell_1$-regularized logistic regression with {\oldnew seven} different calibration schemes for the tuning parameter. Depicted are the results for $p \in\{ 200, 500\}$ and $\kappa= 0.5$.}\label{fig:timef}
\end{figure}

\section{Real data analysis}\label{sec:rd}

In this section, we apply the proposed scheme to biological data. We consider three data sets:
\begin{enumerate}
  \item[a)]{Gene expression data from a leukemia microarray study~\citep{golub1999molecular}. The data comprises $n = 72$ patients; $25$ patients with acute myeloid leukemia  and $47$ patients with acute lymphoblastic leukemia. The predictors are  the expression levels of  $p=7129$ genes. The data is summarized in the \texttt{R}~package \texttt{golubEsets}. The goal is to select the genes whose expression levels discriminate between the two types of leukemia.}
\item[b)]{The above data with the additional preprocessing and filtering  described in~\citep{dudoit2002comparison}. This reduces the number of genes to $p = 3571$. The data is summarized in the \texttt{R}~package \texttt{cancerclass}.}
\item[c)]{Proteomics data from a melanoma study~\citep{mian2005serum}.
 The data comprises $n=205$ patients; $101$ patients with stage~I (moderately severe) melanoma and $104$ patients with stage~IV (very severe) melanoma. The raw data contains the  intensities of $18'856$ mass-charge ($m/z$) values measured in the patients' serum samples. We apply the preprocessing described in~\citep{vasiliu2018penalized}, which results in $p=500$ $m/z$ values, and we subsequently normalize the data. The goal is to select the $m/z$ values whose intensities  discriminate between the two melanoma stages.}
\end{enumerate}

The objective of our method is feature selection. However, since there are no ground truths available for the above applications, we cannot measure feature selection accuracy directly. Instead, we need to infer the method's performance from the number of selected predictors and the prediction accuracy. We generally seek methods that yield a model with a small number of predictors (easy to interpret) and small prediction errors (good fit of the data). Moreover, an increase in prediction accuracy through refitting indicates well-estimated supports, while a deterioration in prediction accuracy through refitting indicates false negatives or false positives. We thus report the model sizes and the prediction errors of Leave-One-Out Cross-Validation without (LOOCV) and with  refitting (LOOCV-refit). Typically, no method is simultaneously dominating in all measures, so that one needs to weight the two aspects according to the objective. For example, the model size is sometimes considered secondary when the goal is prediction, but it is a crucial factor for support recovery.

We apply the four different methods as described in the previous section.
\begin{table}[h!]
\centering
\def~{\hphantom{0}}
\caption{Means and standard deviations of the model sizes and of the misclassification rates for Leave-One-Out Cross-Validation without and with refitting.\vspace{1mm}}{

{
\begin{tabular}{lccc}
 Method&  Model size&LOOCV  & LOOCV-refit \vspace{1mm}\\
 \vspace{2mm}
& \multicolumn{3}{c}{a) Gene expression data with $p=7129$ genes}  \vspace{-1.7mm} \\
 Testing& ~4.35 (1.36) & 0.153 (0.362)  & 0.111 (0.316) \\
 \BIC& ~5.03 (2.78)  &  0.181 (0.387)& 0.111 (0.316) \\
 \CV& 25.75 (3.25) &0.056 (0.231)  & 0.042 (0.201) \\
\AIC& 20.36 (3.07) & 0.069 (0.256) &0.069 (0.256)\vspace{3mm}\\
 \vspace{2mm}
& \multicolumn{3}{c}{b) Gene expression data with $p=3571$ genes}\vspace{-1.7mm}\\
  Testing& ~4.42 (1.39) & 0.167 (0.375)  & 0.125 (0.333)   \\
 \BIC& ~4.99 (2.73)  &  0.194 (0.399)& 0.139 (0.348) \\
 \CV& 25.28 (2.91) &0.056 (0.231)  & 0.069 (0.256) \\
\AIC& 20.17 (3.41)  & 0.083 (0.278) &0.056 (0.231)\vspace{3mm}\\
 \vspace{2mm}
& \multicolumn{3}{c}{c) Proteomics data with $p=500$ $m/z$ values} \vspace{-1.7mm}\\
  Testing& ~1.00 (0.00) & 0.205 (0.405)  & 0.205 (0.405)  \\
 \BIC&
 13.84 (1.26)  &  0.117 (0.322)& 0.117 (0.322)  \\
 \CV& 23.52 (3.22) &0.117 (0.322) & 0.195 (0.397) \\
\AIC& 26.86 (5.38) & 0.122 (0.328) &0.185 (0.390)
\end{tabular}}
}
\label{tab:data}
\end{table}
The results are summarized in Table~\ref{tab:data}. We observe that  the methods form two clusters: On the one hand, \CV\ and \AIC\ provide the most accurate predictions. On the other hand, \BIC\ and the testing-based approach select considerably smaller models and show a larger increase in accuracy after refitting. This is expected, in view of \CV\ and \AIC\ being designed for prediction, and \BIC\ and testing being designed for feature selection. For an ``in-cluster'' comparison, we focus at the testing-based method and \BIC. In the first two data examples, the testing-based method is  dominating \BIC, because it provides more accurate prediction with smaller models. In the third example, \BIC\ is more accurate in prediction, but the testing-based approach provides reasonable prediction (compare especially with \CV\ and \AIC\ after refitting) with only one variable.

\section{Discussion}\label{sec:dis}

We have introduced a scheme for the calibration of $\ell_1$-penalized likelihood for feature selection in logistic regression.  A distinctive feature of the approach  are its theoretical guarantees. Indeed, the new method satisfies optimal finite sample bounds, while for existing methods, the available theory is limited to asymptotic results - or there is no theory at all. Given that in applications, sample sizes are always finite, only  finite sample  theory can provide concrete guidance for practitioners. In addition to the theory, the scheme is easy to implement, computationally efficient, and competitive in simulations and real data applications.

Besides being of direct theoretical and practical value for logistic regression, our contribution also provides new insights into the testing ideas that have been developed in~\cite{chichignoud2016practical}. In particular, we think that its successful use in non-linear regression demonstrates the general potential of testing for tuning parameter calibration and is expected to spark further studies of testing ideas in other modeling frameworks.

A topic for further research are the design assumptions. The focus of this paper is feature selection, where strict assumptions on the correlations in $X$ cannot be avoided. However, it would be interesting to extend our approach  to tasks that are less sensitive to correlations, such as $\ell_2$-estimation and prediction.

\revision{
Another question is whether $\ell_1$-penalization is the right starting point.
In this paper, we introduce an approach to calibrate---in a sense optimally---$\ell_1$-penalized likelihood,
accepting all benefits just as well as all limitations of $\ell_1$-regularization.
It could well be that some theoretical limitations, such as the strict conditions on the design, could be alleviated by starting with a different family of estimators,
or different families could improve computational speed or practical performance.
For example, one could consider estimators with non-convex regularizers, such as SCAD \citep{fan2001variable} and MCP \citep{zhang2010nearly}, or adaptive lasso-type regularizers \citep{zou2006adaptive}, all of which have been shown to reduce bias if (arguably very stringent) conditions on the design matrix are met.
Comparing among different types of estimators is beyond the scope of our paper,
but we stress that our scheme principally applies to any family of estimators:
all it needs is a suitable oracle inequality.
So while this paper only discusses the calibration of $\ell_1$-regularized likelihood,
extensions to other estimators (that might suit a given application much better) appear to be in close reach.
}

\revision{This question is closely related to the choice of~$\constA$.
We found that $\constA=6$ leads to excellent empirical results across a wide range of settings,
while our theoretical justification for this choice requires strong assumptions on the design and the parameter vector.
It would be of interest to broaden the scope of the theoretical justification as well as to theorize the choice of $\constA$ for the application of our calibration method to estimators beyond the~$\ell_1$-regularized likelihood.
}

\section*{Acknowledgements}

{\oldnew We thank the editor, associate editor, and anonymous referees for their comments and suggestions, which helped us considerably in revising the paper.} Wei Li acknowledges support from the China Scholarship Council. We thank Haiyun Jin for his valuable contributions to the implementation of our method, and we thank Micha\"el Chichignoud and Xiao-Hua Zhou for the inspiring discussions and the insightful remarks.

\section*{References}

\bibliography{BIBTeX}

\appendix


\section{}
In this appendix, we provide proofs for our theoretical claims, and we present additional simulation results.
\subsection{Proofs for the theoretical claims}

We provide here proofs for Theorems~\ref{Thm:SR} and~\ref{Thm:AV}.
Throughout this section, we write $\rho(u, v)=\exp(u^\T v)/(1+\exp(u^\T v))$ for vectors $u,v$ of the same length. For ease of notation, we will suppress the subscript $\tuning$ at most instances.

The key quantities in the proofs are two vectors     $\primal=(\prims^\T,\primsc^\T)^\T\in\R^p$ and $\dual=(\duals^\T,\dualsc^\T)^\T \in \R^p$ constructed as follows:

    \begin{enumerate}
      \item define the primal subvector $\prims \in \R^{s}$ such that
       \begin{equation*}
         \prims\in \argmin_{\theta\in \R^{s}}\bigg\{\sum_{i=1}^n(\log(1+\exp(x_{i,S}^\T
       \theta))-y_ix_{i,S}^\T\theta)/n+\tuning\norm{\theta}_1\bigg\}\,;
       \end{equation*}
\item set $\primsc=0\in \R^{p-s}$;
     \item      define the dual vector $\dual\in \R^{p}$ via its elements
      \begin{equation*}
        \dualj\redef \sum_{i=1}^nX_{ij}(y_i-\rho(x_{i},\primal))/(n\tuning)\qquad\quad (j\in[p])\,.
      \end{equation*}
    \end{enumerate}

    The proofs of the theorems are based on three auxiliary lemmas. Figure~\ref{fig:dependence} depicts the dependencies.

\tikzstyle{block} = [rectangle, draw=black, fill=white, node distance=1.1cm, minimum height=2em, font=\color{black}\sffamily, rounded corners=5pt, text width=5em, text badly centered]
\tikzstyle{cloud} = [ellipse, draw=black, fill=white, node distance=1.1cm, minimum height=2em, font=\color{black}\sffamily,text width=4em, text badly centered]
\tikzstyle{line} = [draw, -latex']

 \begin{figure}[h!]
 \centering
 \begin{tikzpicture}[auto, node distance=1.5cm]
    \node [cloud] (A2) {Lemma~\ref{lem:L2norm}};
    \node [cloud, right = of A2] (A1) {Lemma~\ref{lem:PDW}};
    \coordinate (Middle) at ($(A2)!0.5!(A1)$);
     \node [cloud, below = of Middle] (A3) {Lemma~\ref{lem:Linfnorm}};
     \node [block, right = of A3] (Thm1) {Theorem~\ref{Thm:SR}};
     \node [block, right = of Thm1] (Thm2) {Theorem~\ref{Thm:AV}};
    \path [line] (A2)  -- (A3);
    \path [line] (A1) -- (Thm1);
     \path [line] (A3) -- (Thm1);
    \path [line] (Thm1) -- (Thm2);
\end{tikzpicture}
\caption{Dependencies among the lemmas and theorems. For example, the arrow between Lemmas~\ref{lem:L2norm} and~\ref{lem:Linfnorm} indicates that Lemma~\ref{lem:Linfnorm} relies on Lemma~\ref{lem:L2norm}.}
\label{fig:dependence}
\end{figure}
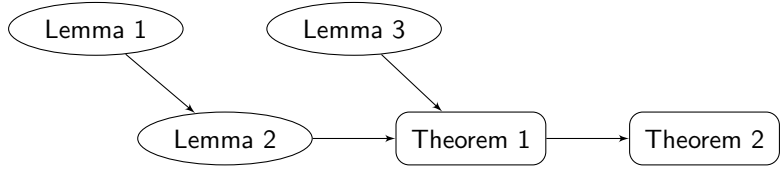
\begin{lemma}[$\ell_{2}$-bound for the primal subvector]
\label{lem:L2norm}
  If $\tuning \leq c_{\min}^2/ (10 sc_{\max})$
   and $\norm{X^\T\varepsilon/n}_{\infty} \leq \tuning/4$,
    then
  \[
  \norm{\prims-\trueparas}_2\leq \frac{5\tuning \surd{s}}{c_{\min}}\,.
  \]
\end{lemma}
\noindent The proof of this lemma follows along the same lines as the proof of Lemma~3 in~\citep{ravikumar2010high}.

\vspace{3mm}
Lemma~\ref{lem:L2norm} implies that the primal subvector $\prims$ is $\ell_2$-consistent, which enables us to develop a Taylor series expansion of $\rho(x_{i,S},\prims)$ at $\trueparas$ according to
\begin{equation}\label{eqn:mvt}
  \rho(x_{i,S},\prims)-\rho(x_{i,S},\trueparas)=\derivafun(x_{i,S},\trueparas) x_{i,S}^\T (\prims-\trueparas)+r_i
\end{equation}
with remainder term
\[r_i=(\prims-\trueparas)^\T \int_0^1(\nabla^2\rho(x_{i,S}, \ \prims+t(\prims-\trueparas)))(1-t)\d t \ (\prims-\trueparas)\,.\]
The derivative reads explicitly
\begin{equation}\label{eqn:ri}
 \nabla^2\rho(x_{i,S},\ \prims+t(\prims-\trueparas))=\xi_i(t) x_{i,S} x_{i,S}^\T\,,
\end{equation}
where $\xi_i(t)=\exp(\eta_i(t))(1-\exp(\eta_i(t)))/
(1+\exp(\eta_i(t)))^3$ and $\eta_i(t)=x_{i,S}^\T (\prims+t(\prims-\trueparas))$
for $t\in[0,1]$ and $i\in [n]$. Summarizing $\res_1,\dots,\res_n$ from Equation~\eqref{eqn:mvt}  in the vector $\res=(r_1,\ldots,r_n)^{\T}$,
we can now state the following result.

\begin{lemma}[$\ell_{\infty}$-bound for the remainder term]
\label{lem:Linfnorm}
  If $\lambda \leq \gamma c_{\min}^2/(100(2-\gamma)sc_{\max})$ and $\norm{X^\T \varepsilon/n}_{\infty} \leq \tuning/4$, then
  \[
  \norm{X^\T \res/n}_{\infty}\leq \frac{\tuning \gamma}{4(2-\gamma)}\,.
  \]
\end{lemma}
\begin{proof}\enspace
The proof follows readily from Lemma~\ref{lem:L2norm}.
To see this, note that because $|X_{ij}|\leq 1$ for all $i\in [n]$ and $j\in [p]$, it holds that
  \begin{equation*}
    \begin{aligned}
      |\res^{\T} x^j/n|=|\sum_{i=1}^n X_{ij}r_i/n|\leq \sum_{i=1}^n|X_{ij}
      ||r_i|/n\leq \sum_{i=1}^n|r_i|/n
      \end{aligned}
  \end{equation*}
    for all $j\in[p]$. By the closed form of $r_i$ in Equation~\eqref{eqn:ri}, it holds that
    $|\xi_i(t)|\leq 1$ for all $t\in[0,1]$, and since $\primsc=\trueparasc=0$, we then get
    \begin{equation}\label{eqn:upperboundL2norm}
    \begin{aligned}
      |\res^{\T} x^j/n|&\leq (\prims-\trueparas)^\T\Big(\sum_{i=1}^n x_{i,S}x_{i,S}^\T/n\Big)(\prims-\trueparas)
      \\
      &=(\prims-\trueparas)^\T (X_{S}^\T X_{S}/n)(\prims-\trueparas)\\
      &=(\primal - \truepara)^\T (X^\T X/n)(\primal-\truepara)\\
      &\leq  c_{\max} \norm{\prims-\trueparas}_2^2\,.
    \end{aligned}
  \end{equation}
 Moreover,  because $\lambda \leq \gamma c_{\min}^2/(100(2-\gamma)sc_{\max})\leq  c_{\min}^2/(10sc_{\max})$ and $\norm{X^\T \varepsilon/n}_{\infty} \leq \tuning/4$, the assumptions of Lemma~\ref{lem:L2norm} are satisfied. Combining this lemma with Equation~\eqref{eqn:upperboundL2norm} yields
  \begin{equation*}
    \begin{aligned}
      |\res^{\T} x^j/n|&\leq \frac{25\tuning^2 sc_{\max}}{c_{\min}^2}
      \leq \frac{\tuning \gamma}{4(2-\gamma)}
    \end{aligned}
  \end{equation*}
  for each $j\in [p]$. Thus, $\norm{X^\T \res/n}_{\infty}\leq \tuning \gamma/(4(2-\gamma))$ as desired.
\end{proof}

\vspace{3mm}

\begin{lemma}[Primal dual witness construction]\label{lem:PDW}
  The pair $(\primal, \dual)$ defined above satisfies the following three properties:
    \begin{enumerate}
      \item[(i)] It holds that $ \duals\in \partial \norm{\prims}_1;$
      \item[(ii)] If $\norm{\dualsc}_{\infty}<1$, then any solution $\estipara$ to the problem~\eqref{eqn:SLR} satisfies \textnormal{supp}$(\estipara)\subset S$;
      \item[(iii)] Under Assumption~\ref{Ass:BEC} and $\norm{\dualsc}_{\infty}<1$, the solution $\estipara$ is unique, and $\estipara = \primal\redef(\prims^\top,\primsc^\top)^\top\in \R^p$.
    \end{enumerate}
\end{lemma}

\begin{proof}\enspace We conduct the proof in three steps in correspondence with the three claims. 

  \textbf{Step 1:}\enspace
   We show that if $\norm{\dualsc}_{\infty}\leq 1$, the pair $(\primal,\dual)$ satisfies the KKT conditions, that is, $\dual\in\partial \norm{\primal}_1$ and
  \[ -\sum_{i=1}^nX_{ij}\left(y_i-\rho(x_{i},\primal)\right)/n+\tuning \dualj=0\]
  for $j \in [p]$.
  By 1. in the construction at the beginning of this section,
    there is a $\kappa\in \partial\norm{\prims}_1\subset \R^{s}$ such that
  \[ -\sum_{i=1}^nX_{ij}\left(y_i-\rho(x_{i,S},\prims)\right)/n+\tuning \kappa_j=0\]
  for $j\in S$.
  Hence, with $\primsc=0$ in 2. and the definition of $\dualj$ in 3.,
  \[ \dualj=\sum_{i=1}^nX_{ij}\left(y_i-\rho(x_{i,S},\prims)\right)/(n\tuning)=\kappa_j\]
  for $j \in S$, that is, $\duals\in \partial \norm{\prims}_1$ as desired.

\textbf{Step 2:}\enspace  We now show that \textnormal{supp}$(\estipara)\subset S$ for all $\estipara\in \R^p$ that satisfy
  \[ \estipara\in \argmin_{\beta\in \R^p}\{L(\para)+\tuning\norm{\beta}_1\}\,.\]
 In view of the condition $\norm{\dualsc}_{\infty}<1$ and of Step~1, the pair $(\primal,\dual)$ satisfies the KKT conditions for the above problem and thus, $\primal$ is a minimizer of the objective function. Consequently,
  \[ L(\estipara)+\tuning\norm{\estipara}_1=L(\primal)+\tuning\norm{\primal}_1\,.\]
  Since,  $\dual\in\partial\norm{\primal}_1$ by Step~1, it holds that $\norm{\primal}_1=\langle \dual,\primal\rangle.$ Plugging this into the previous display yields
  \[ L(\estipara)+\tuning\norm{\estipara}_1=L(\primal)+\tuning\langle \dual,\primal\rangle\,.\]
  We can now subtract $\tuning\langle \dual, \estipara\rangle$ on both sides to obtain
  \[ L(\estipara)+\tuning\norm{\estipara}_1-\tuning\langle \dual, \estipara\rangle=L(\primal)+\tuning\langle \dual,\primal-\estipara\rangle\,.\]
  By 2. and 3. in the above construction, it holds that $\tuning\dual=-L'(\primal)$, where $L'(\cdot)$ denotes the derivative  of~$L(\cdot)$.
  Thus, we can further deduce
  \[ \tuning\norm{\estipara}_1-\tuning\langle \dual,\estipara \rangle=L(\primal)-\langle L'(\primal),\primal-\estipara\rangle-L(\estipara)\,.\]
  Because the Hessian of $L(\para)$ is a non-negative matrix, $L(\cdot)$ is a convex function. It holds that
  \[L(\estipara)\geq L(\primal)+\langle L'(\primal),\estipara-\primal\rangle\,.\]
 Combining the two displays yields
  \[ \tuning \norm{\estipara}_1-\tuning\langle \dual, \estipara\rangle\leq 0\,,\]
 and dividing by the tuning parameter yields further
  \[ \norm{\estipara}_1\leq \langle \dual, \estipara\rangle\,.\]
  However, by H\"{o}lder's inequality and $\norm{\dual}_{\infty}\leq 1$, it holds that
  \[ \norm{\estipara}_1\geq \langle \dual,\estipara\rangle\,.\]
  Consequently,
  \[ \norm{\estipara}_1=\langle \dual,\estipara\rangle\,.\]
  In view of the condition $\norm{\dualsc}_{\infty}<1$, this can only be true if $\estiparaj=0$ for all $j\in {S}^c$. This completes the proof of Step~2.

  \textbf{Step 3:}\enspace We now show that $\estipara=\primal$.
  From Step~2, we deduce that $\estipara=(\estiparas^\T,0)^\T$ with
  \[ \estiparas\in \argmin_{\theta\in \R^{s}}\bigg\{\sum_{i=1}^n(\log(1+\exp(x_{i,S}^\T
       \theta))-y_ix_{i,S}^\T \theta)/n+\tuning\norm{\theta}_1\bigg\}\,. \]
  Moreover, since the minimal eigenvalue of $X_{S}^\top\deriva X_{S}/n$ is larger than zero by Assumption~\ref{Ass:BEC}, this problem has a unique solution.  Combining this with 1. in the construction at the beginning of this section yields  $\estiparas=\prims$, that is, $\estipara = \primal$.
\end{proof}
\noindent {\it Proof of Theorem~\ref{Thm:SR}.}\enspace

We conduct the proof in two steps. The first step is to show that $\textnormal{supp}(\estipara)\subset S$ and that~$\estipara$ is the unique solution of the problem~\eqref{eqn:SLR}. The second step is to show the $\ell_{\infty}$-bound and the result on support recovery.

\textbf{Step 1:}\enspace We first show that $\textnormal{supp}(\estipara)\subset S$ and that $\estipara$ is the unique solution. This result holds true if the primal-dual pair $(\primal,\dual)\in \R^p\times \R^p$ constructed as in Lemma~\ref{lem:PDW} satisfies
  $ \norm{\dualsc}_{\infty}<1$. To show the latter inequality, we use the definition of $\varepsilon$ and Equation~\eqref{eqn:mvt} to rewrite $3.$ in the construction above as
  \[ \sum_{i=1}^n X_{ij} \derivafun(x_{i,S},\trueparas) x_{i,S}^\T (\prims-\trueparas)/n-\sum_{i=1}^n X_{ij}(\varepsilon_i-r_i)/n+\tuning \dualj=0 \qquad\quad (j\in[p])\,.\]
 Because $\derivafun(x_{i,S},\trueparas)=\derivafun(x_{i},\truepara)$ for each $i\in [n]$, we can  put the above display in the matrix form
  \begin{equation*}
   (X^\T \deriva X/n)
   \left(\begin{matrix}
     \trueparas-\prims\\
     0
     \end{matrix}\right)+
     X^\T (\varepsilon-\res)/n -\tuning \left(
     \begin{matrix}
       \duals\\
       \dualsc
     \end{matrix}\right)=0\,,
   \end{equation*}
   and then in the block matrix form
   \begin{equation*}\label{eqn:BM}
     n^{-1}\left(
     \begin{matrix}
       X_{S}^\T \deriva X_{S} & X_{S}^\T \deriva X_{{S}^c}\\
       X_{{S}^c}^\T \deriva X_{S} &X_{{S}^c}^\top \deriva X_{{S}^c}
     \end{matrix}\right)
     \left(\begin{matrix}
     \trueparas-\prims\\
     0
     \end{matrix}\right)+n^{-1}
     \left(\begin{matrix}
       X_{S}^\top(\varepsilon-\res)\\
       X_{{S}^c}^\top(\varepsilon-\res)
     \end{matrix}\right)
     -\tuning\left(
     \begin{matrix}
      \duals\\
       \dualsc
     \end{matrix}\right)=0\,.
   \end{equation*}
   We now solve this equation for $\tuning\dualsc$ and find
   \begin{equation*}\label{eqn:dual}
    \tuning \dualsc=X_{{S}^c}^\T \deriva X_{S}(\trueparas-\prims)/n+X_{{S}^c}^\T(\varepsilon-\res)/n\,.
    \end{equation*}
    Since the matrix $X_{S}^\T W X_{S}$ is invertible by Assumption~\ref{Ass:MIC}, we can solve the block matrix equation also for $(\trueparas-\prims)/n$ and find
   \begin{equation}\label{eqn:para}
     (\trueparas-\prims)/n=-(X_{S}^\T \deriva X_{S})^{-1}X_{S}^\T(\varepsilon-\res)/n+\tuning(X_{S}^\T\deriva X_{S})^{-1}\duals\,.
   \end{equation}
    Combining the two displays yields
   \begin{equation}\label{eqn:tdc}
   \begin{aligned}
   \tuning \dualsc=&-X_{{S}^c}^\T\deriva X_{S}(X_{S}^\T \deriva X_{S})^{-1}X_{S}^\T(\varepsilon-\res)/n+
   X_{{S}^c}^\T(\varepsilon-\res)/n\\
   & +\tuning X_{{S}^c}^\T \deriva X_{S}(X_{S}^\T \deriva X_{S})^{-1}\duals\,.
   \end{aligned}
   \end{equation}
Taking $\ell_{\infty}$-norms on both sides of Equation~\eqref{eqn:tdc} and using the triangle inequality, we find
\begin{equation*}
  \begin{aligned}
    \norm{\tuning \dualsc}_{\infty} \leq \ & \norm{X_{{S}^c}^\T\deriva X_{S}(X_{S}^\T \deriva X_{S})^{-1}X_{S}^\T(\varepsilon-\res)/n}_{\infty} + \norm{X_{{S}^c}^\T(\varepsilon-\res)/n}_{\infty}
    \\  &+\tuning \norm{X_{{S}^c}^\T\deriva X_{S}(X_{S}^\T \deriva X_{S})^{-1}\duals}_{\infty}\,.
  \end{aligned}
\end{equation*}
Invoking properties of the induced matrix norms and the $\ell_\infty$-norm and the condition $\norm{\duals}_{\infty}\leq 1$ deduced in Lemma~\ref{lem:PDW}, and rearranging the terms then provide us with
   \begin{equation*}
   \begin{aligned}
   \norm{\tuning \dualsc}_{\infty}
   &\leq \opnorm{X_{{S}^c}^\T \deriva X_{S}(X_{S}^\T \deriva X_{S})^{-1}}_{\infty}\left(\norm{X^\T (\varepsilon-r)/n}_{\infty}+\lambda\right)
   +\norm{X^\T (\varepsilon-r)/n}_{\infty}\,.
   \end{aligned}
   \end{equation*}
   Next, we divide by $\tuning$ on both sides, apply Assumption~\ref{Ass:MIC}, use the triangle inequality, and rearrange the terms again to find
   \begin{equation*}
     \begin{aligned}
       \norm{\dualsc}_{\infty}&\leq(1-\gamma)+\frac{2-\gamma}{\tuning}(\norm{X^\T \varepsilon/n
       }_{\infty}+\norm{X^\T r/n}_{\infty})\,.
     \end{aligned}
   \end{equation*}
   By the definition of $\mathcal{T}_{\tuning}$, it holds that $\norm{X^\T\varepsilon/n}_{\infty}\leq \tuning \gamma/(4(2-\gamma))$, which is equivalent to
   \begin{equation*}
   \frac{2-\gamma}{\tuning}\norm{X^\T \varepsilon/n}_{\infty}\leq \frac{\gamma}{4}\,.
   \end{equation*}
   Since $\gamma\in(0,1]$, the condition $\norm{X^\T \varepsilon/n}_{\infty} \leq \tuning/4$ in Lemma~\ref{lem:Linfnorm} is satisfied on the event $\mathcal{T}_{\tuning}$. Combining this with the assumption $\tuning \leq \gamma c_{\min}^2/(100(2-\gamma)sc_{\max})$ implies that $\norm{X^\T r/n}_{\infty}\leq \tuning\gamma/(4(2-\gamma))$, see Lemma~\ref{lem:Linfnorm}. Thus,
   \begin{equation*}
     \norm{\dualsc}_{\infty}\leq(1-\gamma)+\frac{\gamma}{4}+\frac{2-\gamma}{\tuning}\norm{X^\T r/n}_{\infty}\leq(1-\gamma)+\frac{\gamma}{4}+\frac{\gamma}{4}<1\,.
   \end{equation*}
  We finally invoke Lemma~\ref{lem:PDW} to conclude that $\textnormal{supp}(\estipara)\subset S$ and that $\estipara = \primal\redef (\prims^\T,0^\T)^\T \in \R^p$ is the unique solution of the problem~\eqref{eqn:SLR}, as desired.

   \textbf{Step 2:}\enspace
   To show the $\ell_{\infty}$-bound, we use Equation~\eqref{eqn:para} and $\estiparas=\prims$ from  Step~1 and find
   \begin{equation*}
    \begin{aligned}
    \trueparas-\estiparas&=-(X_{S}^\T \deriva X_{S})^{-1}X_{S}^\T(\varepsilon-r)+\tuning n(X_{S}^\T \deriva X_{S})^{-1}\duals \\
    &=-\left(X_{S}^\T \deriva X_{S}/n\right)^{-1}\left(X_{S}^\T(\varepsilon-r)/n\right)+\tuning \left(X_{S}^\T \deriva X_{S}/n\right)^{-1}\duals\,.
    \end{aligned}
    \end{equation*}
We then find similarly as before
   \begin{equation*}
     \begin{aligned}
       \norm{\trueparas-\estiparas}_{\infty}
       &\leq \opnorm{(X_{S}^\T \deriva X_{S}/n)^{-1}}_{\infty}(\tuning+\norm{X^\T
       (\varepsilon-r)/n}_{\infty})\,.
     \end{aligned}
   \end{equation*}
   By the definition of $\ratio$, we have
   \begin{equation*}
   \begin{aligned}
   \opnorm{(X_{S}^\T \deriva X_{S}/n)^{-1}}_{\infty}\leq \ratio \,\opnorm{(X_{S}^\T \deriva X_{S}/n)^{-1}}_{2}
   \leq \ratio/c_{\min}\,.
   \end{aligned}
   \end{equation*}
   Combining this with the bounds on $\norm{X^\T \varepsilon/n}_{\infty}$ and $\norm{X^\T r/n}_{\infty}$ deduced   in Step~1 yields
    \[
   \norm{\trueparas-\estiparas}_{\infty}\leq \frac{\ratio }{c_{\min}}\Big(\tuning + \frac{\tuning\gamma}{4(2-\gamma)}+\frac{\tuning\gamma}{4(2-\gamma)}\Big)\leq 1.5\ratio\tuning/c_{\min}\,.
   \]
   Since $\textnormal{supp}(\estipara)\subset S$ by Step~1, the above display implies that
   \[
   \norm{\truepara-\estipara}_{\infty}\leq 1.5\ratio\tuning/c_{\min}\,.
   \]
   Consequently, $\textnormal{supp}(\estipara)= S$ as long as $\min_{j \in S}|\trueparaj|>1.5\ratio\tuning/c_{\min}$. This concludes the proof.
   \QEDB

\vspace{3mm}\noindent {\it Proof of Theorem~\ref{Thm:AV}.}\enspace

The proof is conducted in three steps. The first step is to show the bound on $\hat{\tuning}$, the second step is to show the bound on the sup-norm error, and the last step is to  show that $\hat{S}\supset S$. To begin with,
 we define the event
 \[ \mathcal{T}_{\delta}^*=\left\{\frac{4(2-\gamma)}{n\gamma}\norm{X^\T \varepsilon}_{\infty} \leq \tuning_{\delta}^* \right\}.
 \]
 By our definition of the oracle tuning parameter in~\eqref{eqn:AV}, we have
that $\Pr(\mathcal{T}_{\delta}^*) \geq 1 - \delta$.
Thus, it suffices to show that the results hold conditioned on the
event $\mathcal{T}_{\delta}^*$.

\textbf{Step 1:}\enspace  To show that
$\hat{\tuning} \leq \tuning_{\delta}^*$, we proceed by proof by contradiction. If $\hat{\tuning}>\tuning_{\delta}^*$, then the definition of our testing-based calibration implies that there must exist two tuning parameters $\tuning',\tuning''\geq \tuning_{\delta}^*$ such that
\begin{align}\label{eqntwotuning}
\norm{\estipara_{\tuning'}-\estipara_{\tuning''}}_{\infty}>C(\tuning'+\tuning'')\,.
\end{align}
However, because both $\mathcal{T}_{\tuning'}$ and $\mathcal{T}_{\tuning''}$
include $\mathcal{T}_{\delta}^*$, and because $C\geq 1.5\ratio/c_{\min}$, Theorem~\ref{Thm:SR} implies that $\norm{\estipara_{\tuning'}-\truepara}_{\infty}\leq C\tuning'$ and $\norm{\estipara_{\tuning''}-\truepara}_{\infty}\leq C\tuning''$. By applying the triangle inequality, we have
\[
\norm{\estipara_{\tuning'}-\estipara_{\tuning''}}_{\infty}\leq \norm{\estipara_{\tuning'}-\truepara}_{\infty}+\norm{\estipara_{\tuning''}-\truepara}_{\infty}\leq
C(\tuning'+\tuning'')\,.
\]
This upper bound contradicts our earlier conclusion~\eqref{eqntwotuning} and, therefore, yields the desired bound on the tuning parameter.

\textbf{Step 2:}\enspace On the event $\mathcal{T}_{\delta}^*$, we have $\hat{\tuning} \leq \tuning_{\delta}^*$, and so the testing-based method implies that
\[
\norm{\estipara_{\hat{\tuning}}-\estipara_{\tuning_{\delta}^*}}_{\infty}\leq C(\hat{\tuning}+\tuning_{\delta}^*)\leq 2 C\tuning_{\delta}^*\,.
\]
By applying the triangle inequality, we find that
\[
\norm{\estipara_{\hat{\tuning}}-\truepara}_{\infty}\leq \norm{\estipara_{\hat{\tuning}}-\estipara_{\tuning_{\delta}^*}}_{\infty}+
\norm{\estipara_{\tuning_{\delta}^*}-\truepara}_{\infty}\leq
2C\tuning_{\delta}^*+\norm{\estipara_{\tuning_{\delta}^*}-\truepara}_{\infty}\,.
\]
Theorem~\ref{Thm:SR} implies that $\norm{\estipara_{\tuning_{\delta}^*}-\truepara}_{\infty}\leq 1.5\ratio\tuning_{\delta}^*/c_{\min}\leq C\tuning_{\delta}^*$, and combining the pieces  yields the desired sup-norm bound.

\textbf{Step 3:}\enspace Let us finally show that $\hat{S}\supset S$. Suppose $j\in S$, then by the bound on the sup-norm error that we deduce in Step~2, we have
\[
|(\estipara_{\hat{\tuning}})_j|\geq |\trueparaj| - 3C\tuning_{\delta}^*\,.
\]
In view of the condition $\min_{j\in S}|\trueparaj|>6C\tuning_{\delta}^*$ and the definition $\hat{S}=\{j\in [p]:|(\estipara_{\hat{\tuning}})_j|\geq 3C\hat{\tuning}\}$, we conclude that $j\in \hat{S}$, that is, $\hat{S}\supset S$. This completes the proof.
\QEDB
\subsection{Additional simulations}\label{sec:appsim}


In this section, we first present additional simulation results for the settings described in Section~\ref{sec:sim}. Figure~\ref{fig:hdapp} shows the results of Hamming distance for $\kappa\in\{0,0.75\}$ and $s\in\{8,12,15\}$, and Figure~\ref{fig:timeapp} shows corresponding run times for $\kappa=0.75$. The results for more dense signals ($s\in\{20,25\}$) and $\kappa\in\{0,0.25,0.5,0.75\}$ and $p\in\{200,500\}$ are summarized in Figure~\ref{fig:hdappdense} and~\ref{fig:timeappdense}.

Next, we  expand our simulation studies by varying covariance structures for the covariates and manipulating the values of the $\beta^*$'s. As pointed out in Section~\ref{sec:sim}, we generate the row vectors $x_i$ of $X$ independently from Gaussian with mean zero and covariance $\Sigma$. Here, we consider $\Sigma$ as $\Sigma=(\sigma_{ij})$, where $\sigma_{ij}=0.9^{|i-j|}$. We then generate each component of $\beta^*$ in the support set $S$ from $N(\mu_{\beta^*},1)$. The parameter $\mu_{\beta^*}$ is used to control the signal level, and we consider $\mu_{\beta^*}=2.5$ and 5 for low and high signal levels, respectively. Similar to the setting in Section~\ref{sec:sim}, we still set $n=200$, $p\in\{200,500\}$, $s\in\{8,12,15,20,25\}$. All the
results are averaged over 200 replications. Figure~\ref{fig:hdnewsetting} and~\ref{fig:timenewsetting} show the results of Hamming distance and run times for $s\in\{8,12,15\}$, respectively. Figure~\ref{fig:hdnewsettingdense} and~\ref{fig:timenewsettingdense} show the corresponding results for $s\in\{20,25\}$.

\revision{
We finally test the proposed scheme in ultra-high dimensional settings.
We consider $p=3000$ and leave all other parameters as in Section~\ref{sec:sim}, that is, $n=200$, $\kappa\in\{0, 0.25, 0.5, 0.75\}$, and $s\in\{8, 12, 15, 20, 25\}$.
Figure~\ref{fig:hdappultrahigh} shows the Hamming distances for $s\in\{8, 12, 15\}$, and Figure~\ref{fig:timeappultrahigh} shows the corresponding run times.
Figure~\ref{fig:hdappultrahighdense} and Figure~\ref{fig:timeappultrahighdense} show the same for the more dense cases $s\in\{20, 25\}$.
}

\revision{
Figures~\ref{fig:hdapp}~--~\ref{fig:timeappultrahighdense} indicate that our testing-based method rivals or outmatches all alternatives across a very wide range of settings, both in computational time and in statistical accuracy.
Hence, beyond its theoretical guarantees (which none of the alternatives has), the proposed scheme is also a competitor in practice.
}


\begin{figure}
\centering
       \subfigure[~ $p = 200$, $\kappa = 0$]{
  \includegraphics[width=0.45\textwidth]{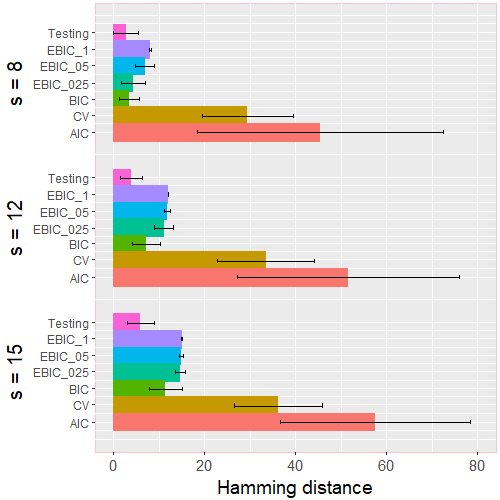}
  }
  ~~~~
\subfigure[~ $p = 200$, $\kappa = 0.75$]{
  \includegraphics[width=0.45\textwidth]{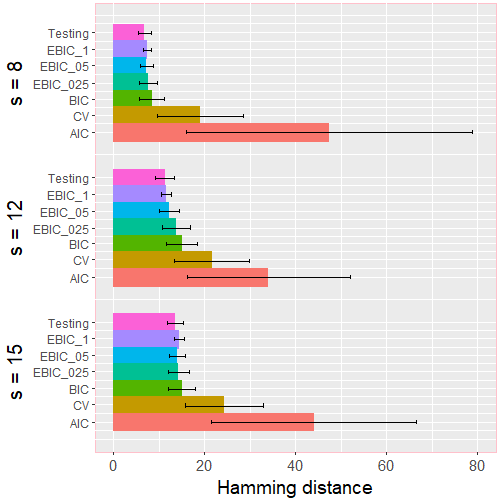}
  }
  \vspace{6mm}
 \subfigure[~ $p = 500$, $\kappa = 0$]{{
  \includegraphics[width=0.45\textwidth]{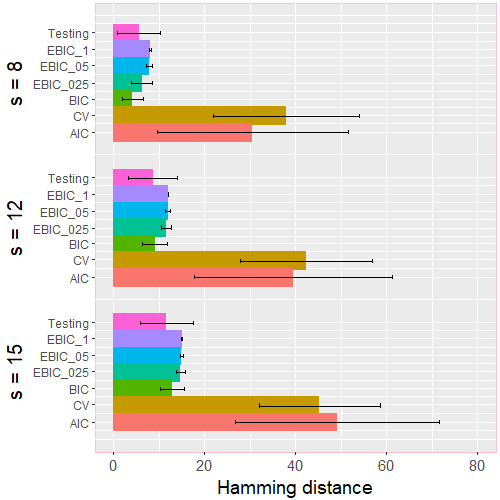}
  }}
  ~~~~
  \subfigure[~ $p = 500$, $\kappa = 0.75$]{{
  \includegraphics[width=0.45\textwidth]{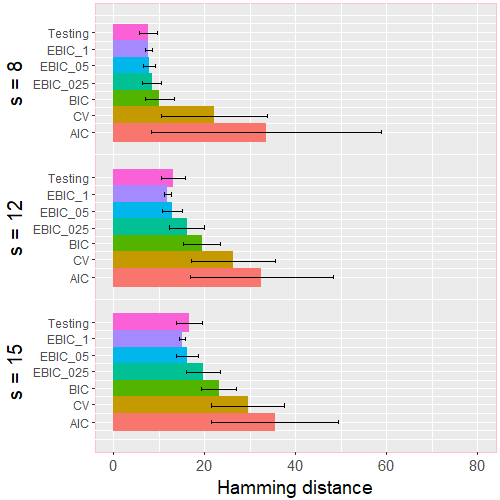}
  }}
  \caption{Variable selection errors of $\ell_1$-regularized logistic regression with  seven different tuning parameter calibration schemes for settings described in Section~\ref{sec:sim}. The $12$ simulation settings differ in the number of variables $p$, correlation $\kappa$, and sparsity level~$s\in\{8,12,15\}$.}\label{fig:hdapp}
\end{figure}

\begin{figure}
\centering
       \subfigure[~ $p = 200$, $\kappa = 0.75$]{
  \includegraphics[width=0.45\textwidth]{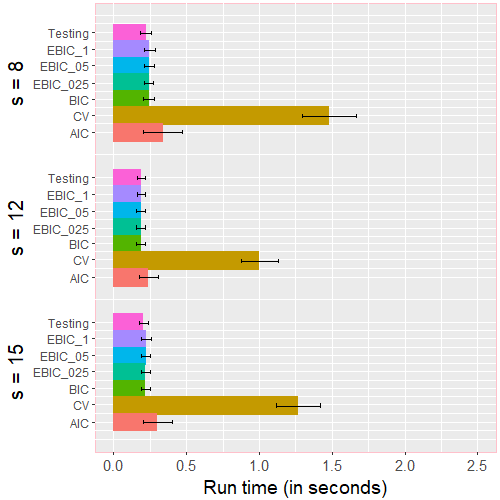}
  }
  ~~~~
\subfigure[~ $p = 500$, $\kappa = 0.75$]{
  \includegraphics[width=0.45\textwidth]{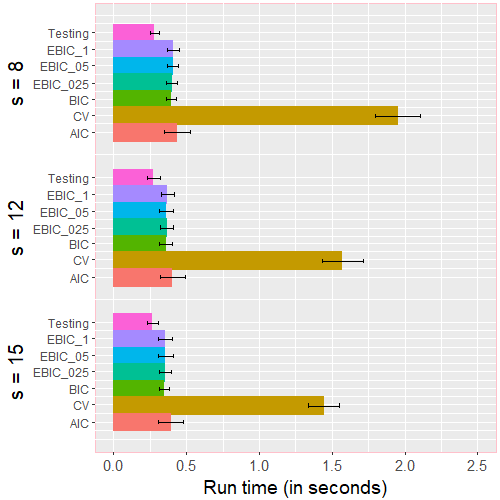}
  }
  \caption{Run times (in seconds) of $\ell_1$-regularized logistic regression with seven different tuning parameter calibration schemes for settings described in Section~\ref{sec:sim}. Depicted are the results for $p \in\{200, 500\}$, $\kappa= 0.75$ and $s\in\{8,12,15\}$.}\label{fig:timeapp}
\end{figure}

\begin{figure}
\centering
\subfigure[~ $p = 200$, $\kappa = 0.25$]{
  \includegraphics[width=0.35\textwidth]{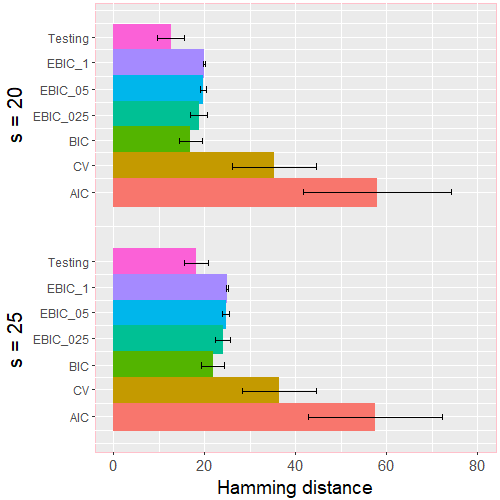}
  }
  ~~~~
\subfigure[~ $p = 200$, $\kappa = 0.5$]{
  \includegraphics[width=0.35\textwidth]{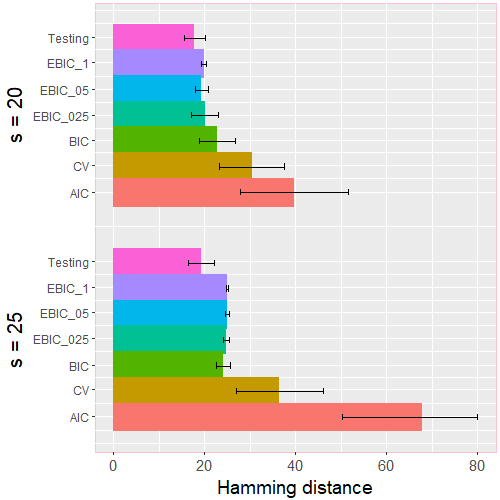}
  }
 \subfigure[~ $p = 500$, $\kappa = 0.25$]{{
  \includegraphics[width=0.35\textwidth]{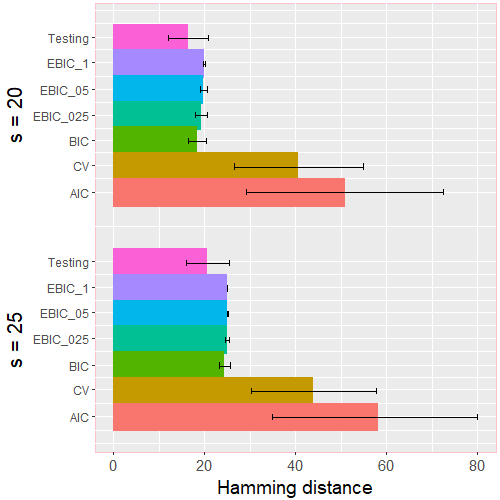}
  }}
  ~~~~
  \subfigure[~ $p = 500$, $\kappa = 0.5$]{{
  \includegraphics[width=0.35\textwidth]{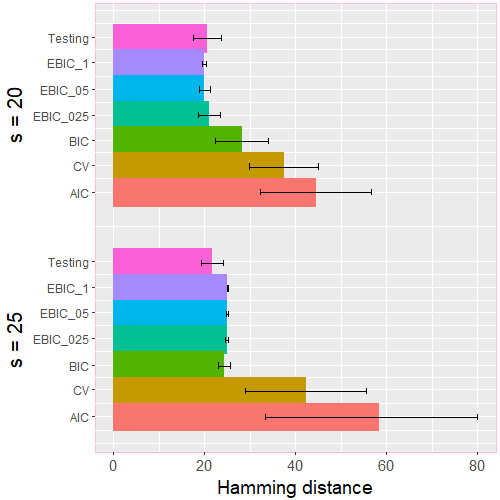}
  }}
       \subfigure[~ $p = 200$, $\kappa = 0$]{
  \includegraphics[width=0.35\textwidth]{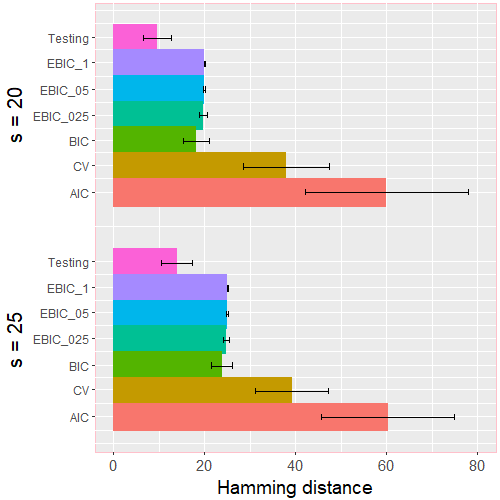}
  }
  ~~~~
\subfigure[~ $p = 200$, $\kappa = 0.75$]{
  \includegraphics[width=0.35\textwidth]{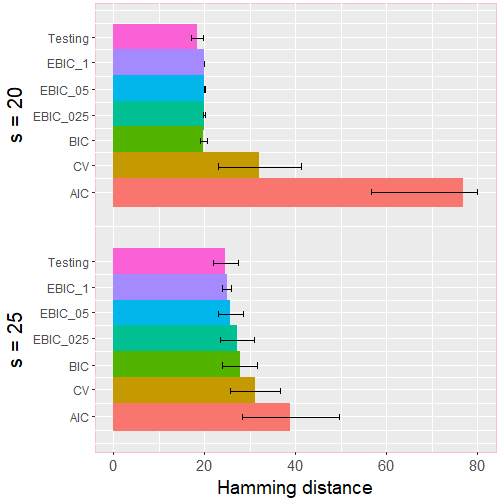}
  }
 \subfigure[~ $p = 500$, $\kappa = 0$]{{
  \includegraphics[width=0.35\textwidth]{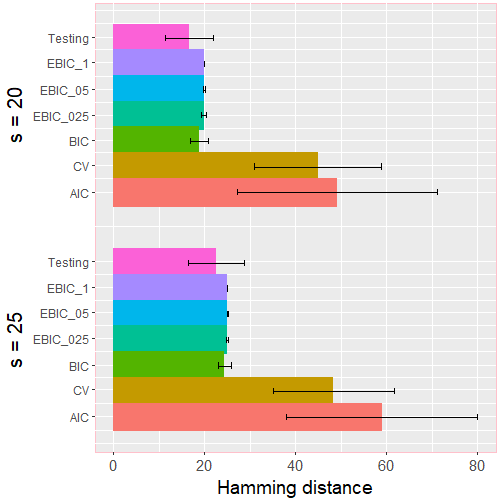}
  }}
  ~~~~
  \subfigure[~ $p = 500$, $\kappa = 0.75$]{{
  \includegraphics[width=0.35\textwidth]{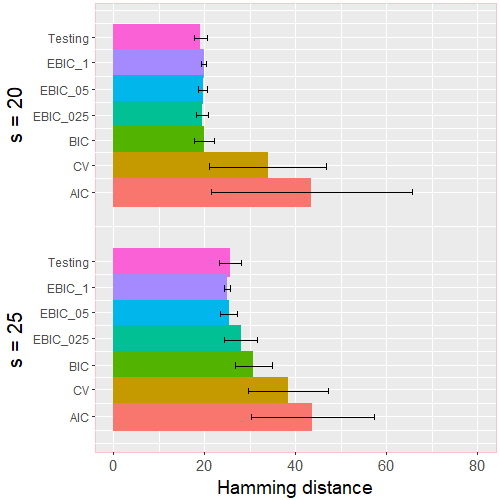}
  }}
  \caption{Variable selection errors of $\ell_1$-regularized logistic regression with seven different tuning parameter calibration schemes for settings described in Section~\ref{sec:sim}. The $24$ simulation settings differ in the number of variables $p$, correlation $\kappa$, and sparsity level~$s\in\{20,25\}$.}\label{fig:hdappdense}
\end{figure}

\begin{figure}
\centering
       \subfigure[~ $p = 200$, $\kappa = 0.75$]{
  \includegraphics[width=0.45\textwidth]{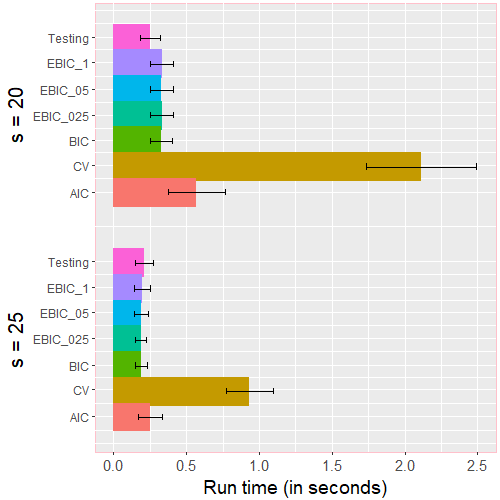}
  }
  ~~~~
\subfigure[~ $p = 500$, $\kappa = 0.75$]{
  \includegraphics[width=0.45\textwidth]{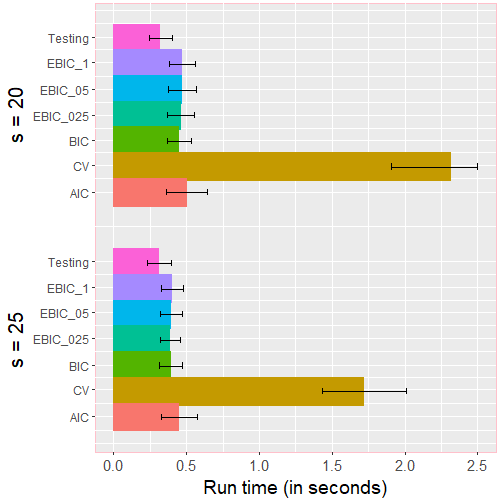}
  }
  \caption{Run times (in seconds) of $\ell_1$-regularized logistic regression with seven different tuning parameter calibration schemes for settings described in Section~\ref{sec:sim}. Depicted are the results for $p \in\{200, 500\}$, $\kappa= 0.75$, and $s\in\{20,25\}$.}\label{fig:timeappdense}
\end{figure}

\begin{figure}
\centering
       \subfigure[~ $p = 200$, $\mu_{\beta^*} = 2.5$]{
  \includegraphics[width=0.45\textwidth]{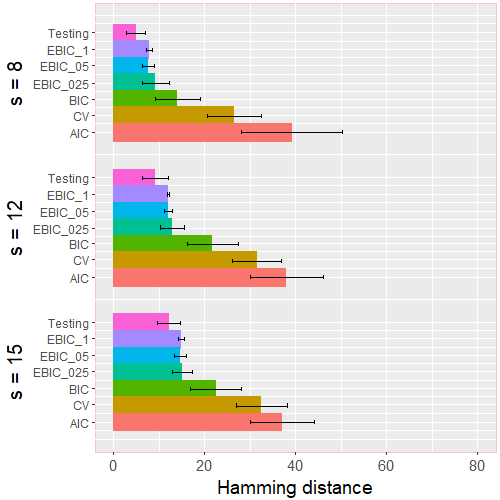}
  }
  ~~~~
\subfigure[~ $p = 200$, $\mu_{\beta^*} = 5$]{
  \includegraphics[width=0.45\textwidth]{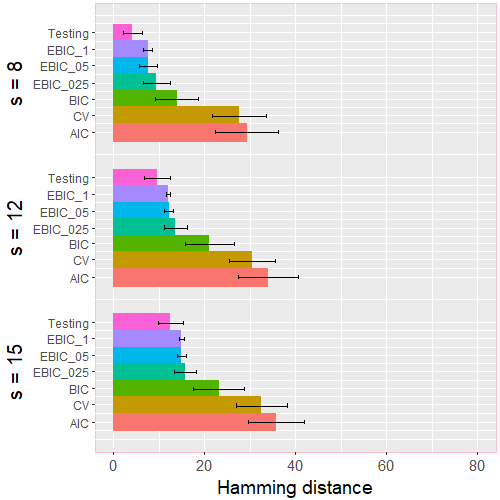}
  }
  \vspace{6mm}
 \subfigure[~ $p = 500$, $\mu_{\beta^*} = 2.5$]{{
  \includegraphics[width=0.45\textwidth]{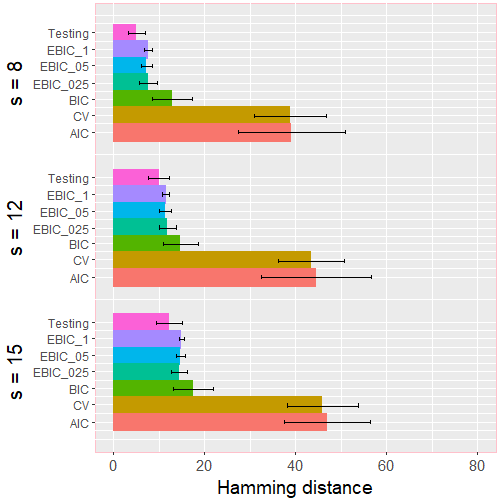}
  }}
  ~~~~
  \subfigure[~ $p = 500$, $\mu_{\beta^*} = 5$]{{
  \includegraphics[width=0.45\textwidth]{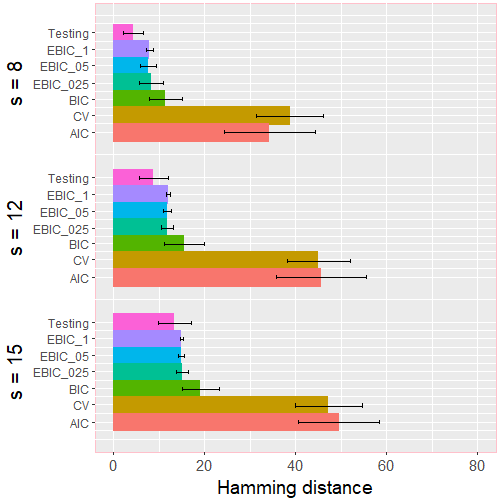}
  }}
  \caption{Variable selection errors of $\ell_1$-regularized logistic regression with seven different tuning parameter calibration schemes for settings described in~\ref{sec:appsim}. The $12$ simulation settings differ in the number of variables $p$, signal level $\mu_{\beta^*}$, and sparsity level~$s\in\{8,12,15\}$.}\label{fig:hdnewsetting}
\end{figure}

\begin{figure}
\centering
       \subfigure[~ $p = 500$, $\mu_{\beta^*} = 2.5$]{
  \includegraphics[width=0.45\textwidth]{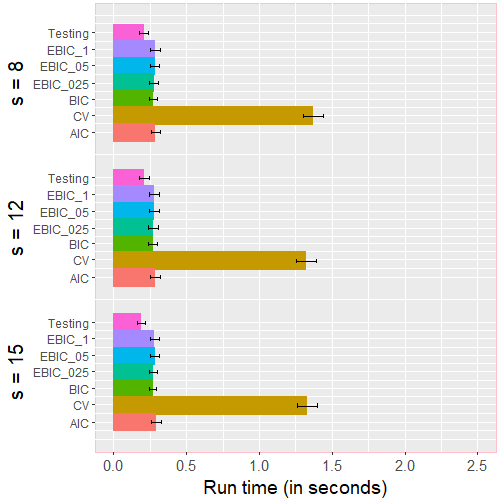}
  }
  ~~~~
\subfigure[~ $p = 500$, $\mu_{\beta^*} = 5$]{
  \includegraphics[width=0.45\textwidth]{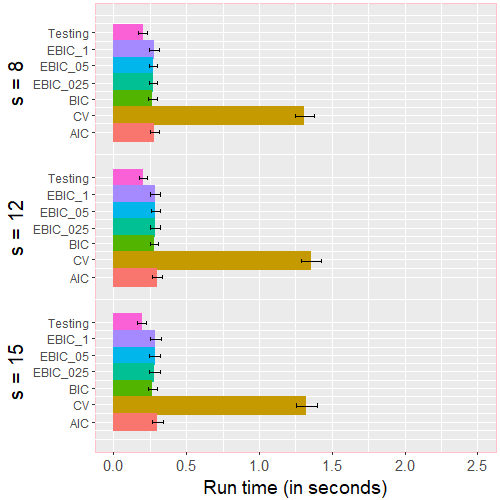}
  }
  \caption{Run times (in seconds) of $\ell_1$-regularized logistic regression with  seven different tuning parameter calibration schemes for settings described in~\ref{sec:appsim}. Depicted are the results for $p =500$, $\mu_{\beta^*}\in\{2.5,5\}$ and $s\in\{8,12,15\}$.}\label{fig:timenewsetting}
\end{figure}

\begin{figure}
\centering
       \subfigure[~ $p = 200$, $\mu_{\beta^*} = 2.5$]{
  \includegraphics[width=0.45\textwidth]{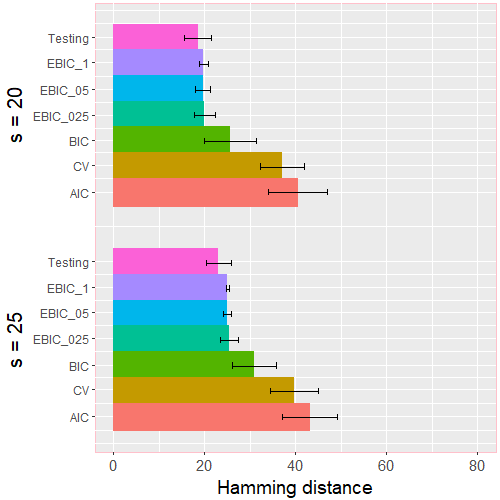}
  }
  ~~~~
\subfigure[~ $p = 200$, $\mu_{\beta^*} = 5$]{
  \includegraphics[width=0.45\textwidth]{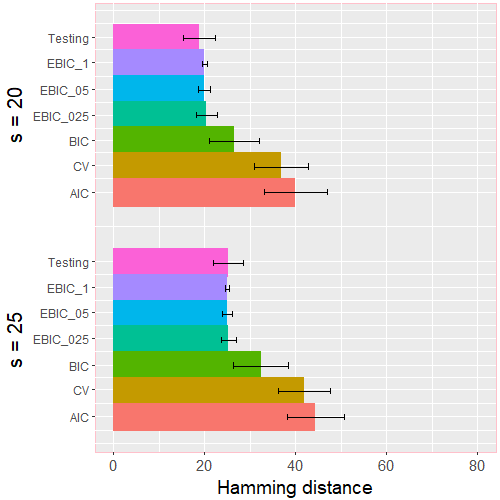}
  }
  \vspace{6mm}
 \subfigure[~ $p = 500$, $\mu_{\beta^*} = 2.5$]{{
  \includegraphics[width=0.45\textwidth]{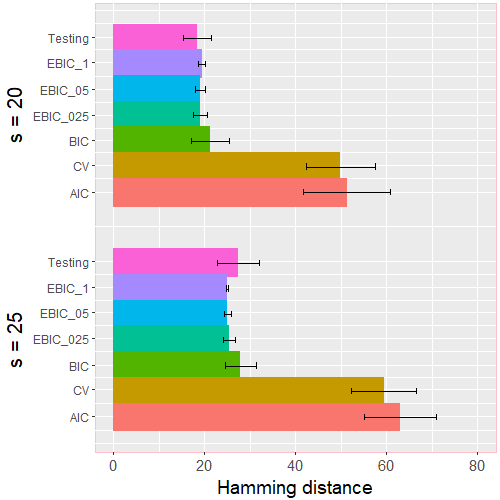}
  }}
  ~~~~
  \subfigure[~ $p = 500$, $\mu_{\beta^*} = 5$]{{
  \includegraphics[width=0.45\textwidth]{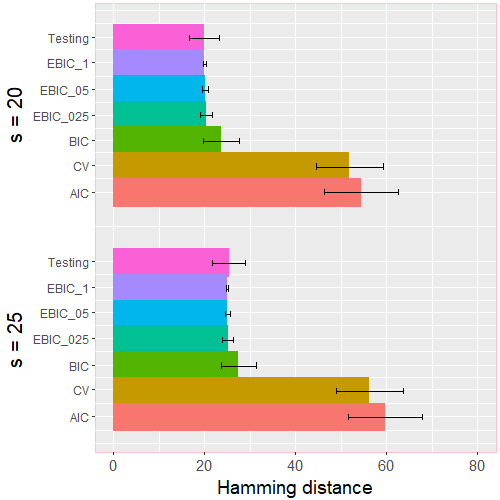}
  }}
  \caption{Variable selection errors of $\ell_1$-regularized logistic regression with seven different tuning parameter calibration schemes for settings described in~\ref{sec:appsim}. The $12$ simulation settings differ in the number of variables $p$, signal levels $\mu_{\beta^*}$, and sparsity level~$s\in\{20,25\}$.}\label{fig:hdnewsettingdense}
\end{figure}

\begin{figure}
\centering
       \subfigure[~ $p = 500$, $\mu_{\beta^*} = 2.5$]{
  \includegraphics[width=0.45\textwidth]{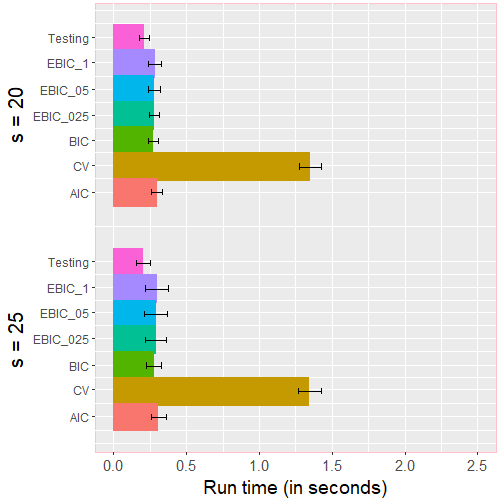}
  }
  ~~~~
\subfigure[~ $p = 500$, $\mu_{\beta^*} = 5$]{
  \includegraphics[width=0.45\textwidth]{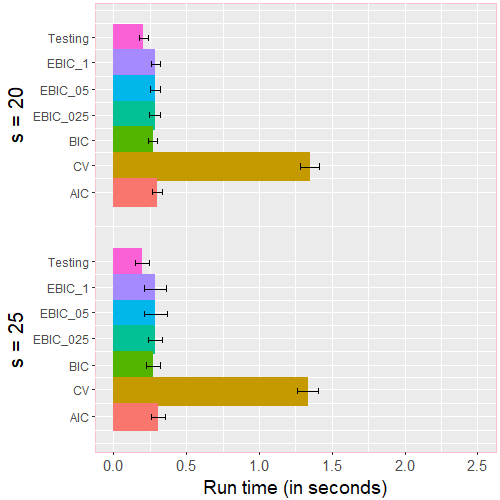}
  }
  \caption{Run times (in seconds) of $\ell_1$-regularized logistic regression with seven different tuning parameter calibration schemes for settings described in~\ref{sec:appsim}. Depicted are the results for $p =500$, $\mu_{\beta^*}\in\{2.5,5\}$ and $s\in\{20,25\}$.}\label{fig:timenewsettingdense}
\end{figure}

\begin{figure}
\centering
       \subfigure[~ $p = 3000$, $\kappa = 0$]{
  \includegraphics[width=0.45\textwidth]{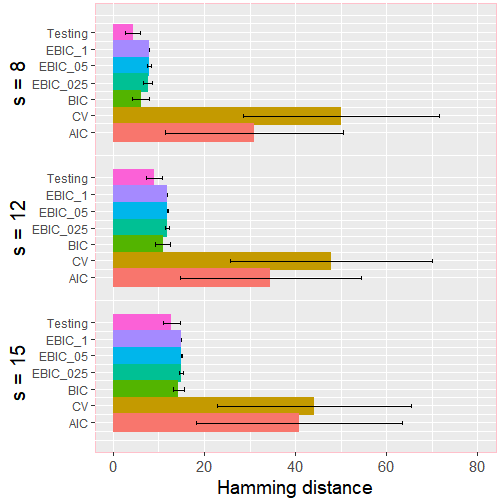}
  }
  ~~~~
\subfigure[~ $p = 3000$, $\kappa = 0.25$]{
  \includegraphics[width=0.45\textwidth]{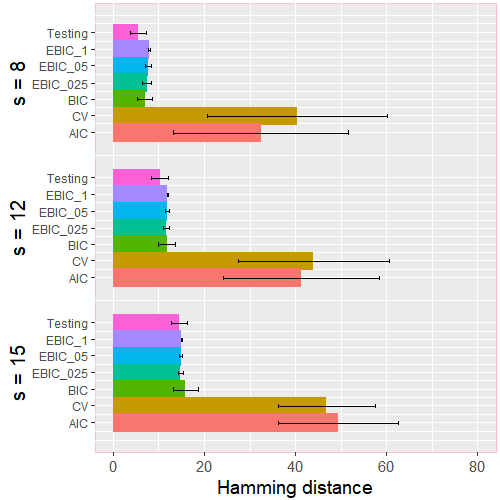}
  }
  \vspace{6mm}
 \subfigure[~ $p = 3000$, $\kappa = 0.5$]{{
  \includegraphics[width=0.45\textwidth]{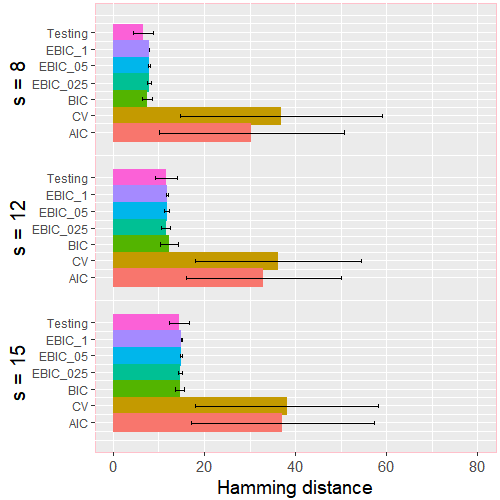}
  }}
  ~~~~
  \subfigure[~ $p = 3000$, $\kappa = 0.75$]{{
  \includegraphics[width=0.45\textwidth]{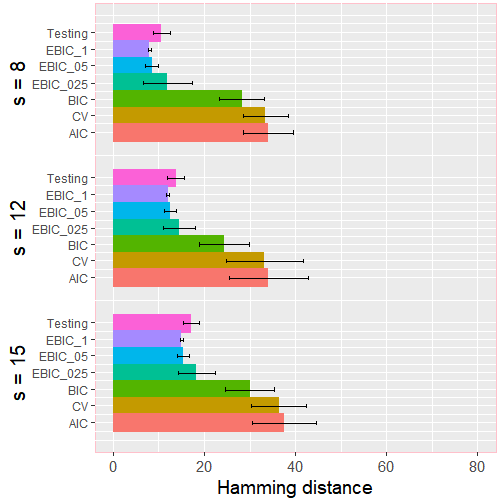}
  }}
  \caption{Variable selection errors of $\ell_1$-regularized logistic regression with  seven different tuning parameter calibration schemes for settings described in Section~\ref{sec:sim}. The $12$ simulation settings share $p=3000$ but differ in the  correlation level~$\kappa\in\{0, 0.25, 0.5, 0.75\}$ and sparsity level~$s\in\{8,12,15\}$.}\label{fig:hdappultrahigh}
\end{figure}

\begin{figure}
\centering
       \subfigure[~ $p = 3000$, $\kappa = 0$]{
  \includegraphics[width=0.45\textwidth]{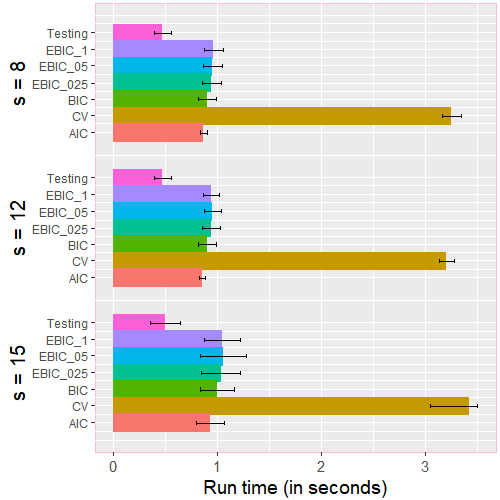}
  }
  ~~~~
\subfigure[~ $p = 3000$, $\kappa = 0.25$]{
  \includegraphics[width=0.45\textwidth]{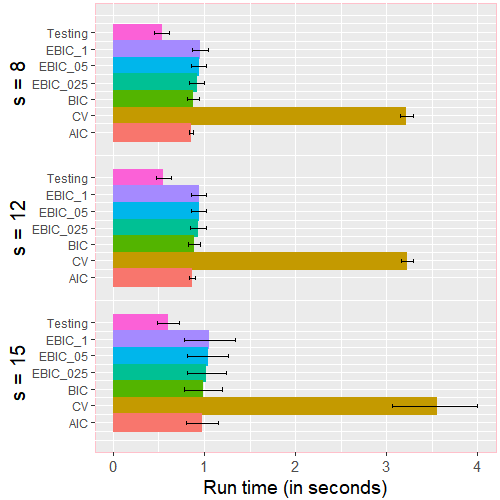}
  }
  \vspace{6mm}
 \subfigure[~ $p = 3000$, $\kappa = 0.5$]{{
  \includegraphics[width=0.45\textwidth]{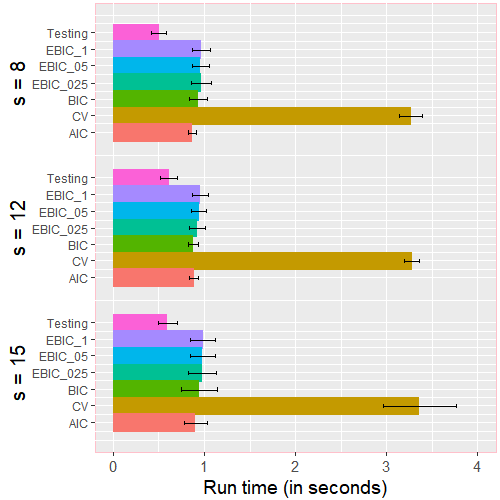}
  }}
  ~~~~
  \subfigure[~ $p = 3000$, $\kappa = 0.75$]{{
  \includegraphics[width=0.45\textwidth]{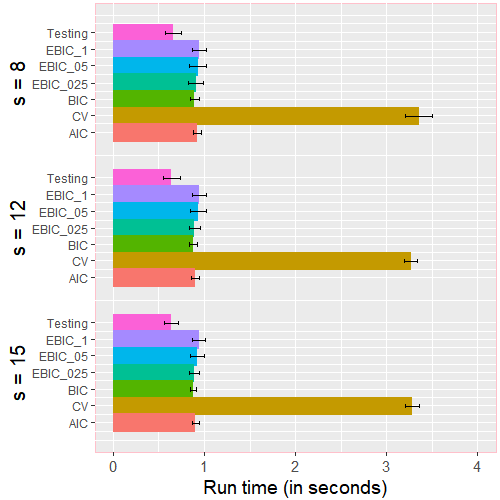}
  }}
  \caption{Run times (in seconds) of $\ell_1$-regularized logistic regression with  seven different tuning parameter calibration schemes for settings described in Section~\ref{sec:sim}. The $12$ simulation settings share $p=3000$ but differ in the  correlation level~$\kappa\in\{0, 0.25, 0.5, 0.75\}$ and sparsity level~$s\in\{8,12,15\}$.}\label{fig:timeappultrahigh}
\end{figure}

\begin{figure}
\centering
       \subfigure[~ $p = 3000$, $\kappa = 0$]{
  \includegraphics[width=0.45\textwidth]{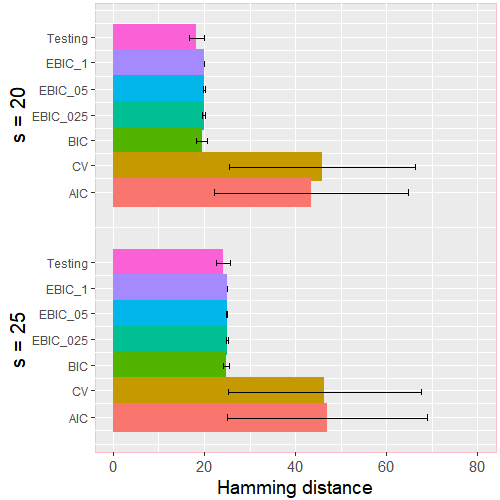}
  }
  ~~~~
\subfigure[~ $p = 3000$, $\kappa = 0.25$]{
  \includegraphics[width=0.45\textwidth]{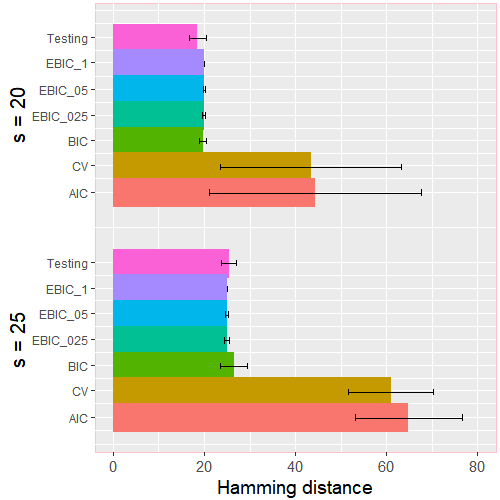}
  }
  \vspace{6mm}
 \subfigure[~ $p = 3000$, $\kappa = 0.5$]{{
  \includegraphics[width=0.45\textwidth]{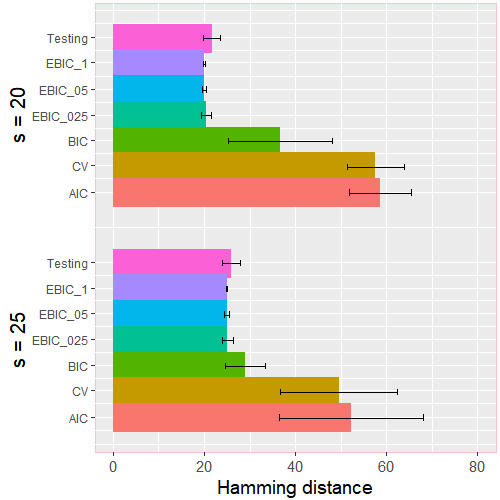}
  }}
  ~~~~
  \subfigure[~ $p = 3000$, $\kappa = 0.75$]{{
  \includegraphics[width=0.45\textwidth]{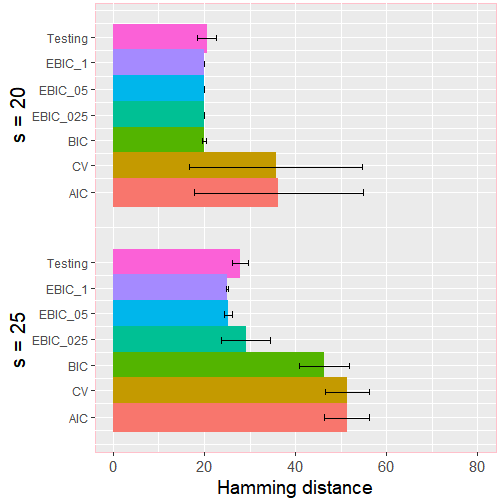}
  }}
  \caption{Variable selection errors of $\ell_1$-regularized logistic regression with  seven different tuning parameter calibration schemes for settings described in Section~\ref{sec:sim}. The $8$ simulation settings share $p=3000$ but differ in the  correlation level~$\kappa\in\{0, 0.25, 0.5, 0.75\}$ and sparsity level~$s\in\{20,25\}$.}\label{fig:hdappultrahighdense}
\end{figure}

\begin{figure}
\centering
       \subfigure[~ $p = 3000$, $\kappa = 0$]{
  \includegraphics[width=0.45\textwidth]{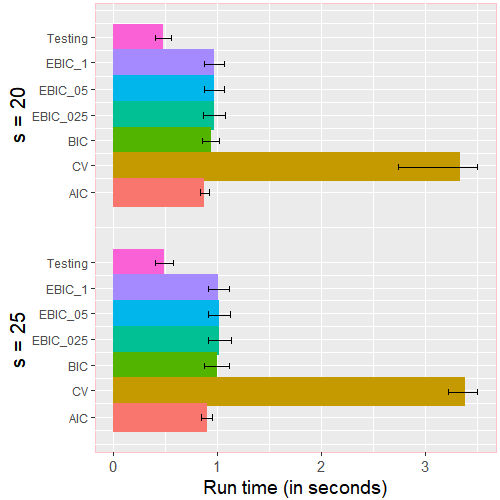}
  }
  ~~~~
\subfigure[~ $p = 3000$, $\kappa = 0.25$]{
  \includegraphics[width=0.45\textwidth]{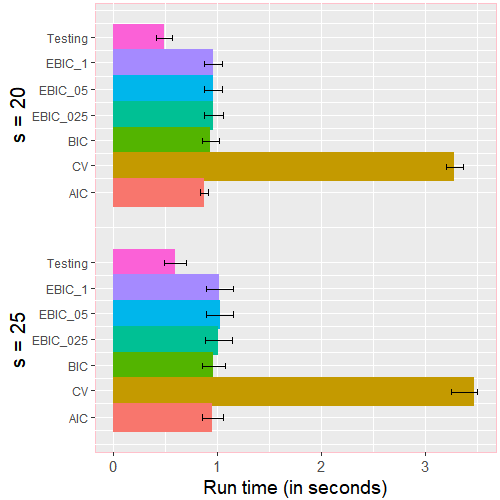}
  }
  \vspace{6mm}
 \subfigure[~ $p = 3000$, $\kappa = 0.5$]{{
  \includegraphics[width=0.45\textwidth]{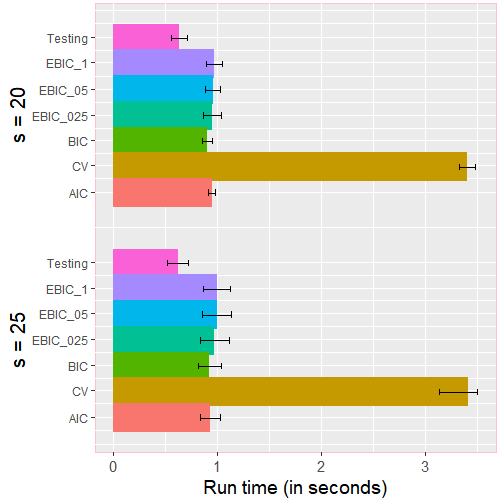}
  }}
  ~~~~
  \subfigure[~ $p = 3000$, $\kappa = 0.75$]{{
  \includegraphics[width=0.45\textwidth]{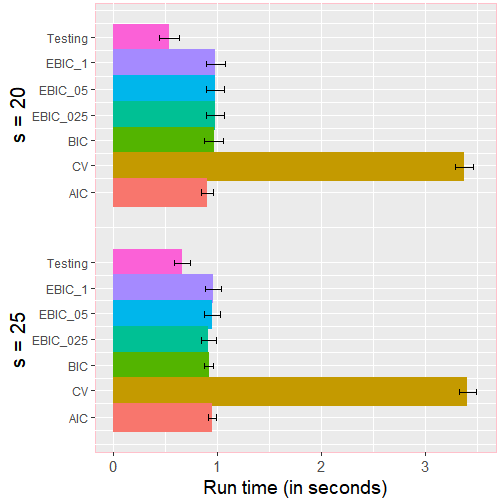}
  }}
  \caption{Run times (in seconds) of $\ell_1$-regularized logistic regression with  seven different tuning parameter calibration schemes for settings described in Section~\ref{sec:sim}. The $8$ simulation settings share $p=3000$, but differ in the  correlation level~$\kappa\in\{0, 0.25, 0.5, 0.75\}$, and sparsity level~$s\in\{20,25\}$.}\label{fig:timeappultrahighdense}
\end{figure}

\end{document}